%% file: cones_manuscript.tex
\documentclass[a4paper,UKenglish,twoside]{article}

\usepackage[numbers]{natbib}
\usepackage{hyperref}

\usepackage{microtype}
\usepackage{graphicx}
\usepackage{mathtools}
\usepackage{amsmath}
\usepackage{amssymb}
\usepackage{amsthm}
\usepackage{refcount}
\usepackage{etex,etoolbox}
\usepackage{verbatim}
\usepackage{mathrsfs}
\usepackage{tikz}
\usepackage{tikz-cd}
\usepackage{pstricks}
\usepackage{pst-node}
\usepackage{enumitem}
\usepackage{xcolor}
\usepackage{hyperref}
\usepackage{enumitem}
\usetikzlibrary{matrix}
\usepackage{bbold}
\usepackage{authblk}
\usepackage{fancyhdr}
\usepackage{datetime}

\makeatletter
\newcommand{\pushright}[1]{\ifmeasuring@#1\else\omit\hfill$\displaystyle#1$\fi\ignorespaces}
\newcommand{\pushleft}[1]{\ifmeasuring@#1\else\omit$\displaystyle#1$\hfill\fi\ignorespaces}
\makeatother

\renewcommand\qedsymbol{$\blacksquare$}

\newcommand{\emptyword}{\varepsilon}
\newcommand{\bra}[1]{\ensuremath{\langle#1|}}
\newcommand{\ket}[1]{\ensuremath{|#1\rangle}}

\newcommand{\ketbra}[2]{\ensuremath{|#1\rangle\!\langle#2|}}

\newcommand{\tr}{\mathrm{tr}}
\newcommand{\C}{\mathscr C}

\newcommand{\vecspan}{\ensuremath{\mathrm{span}}}

\newcommand{\M}{\mathcal{M}}
\newcommand{\uu}{{\bf u}}
\newcommand{\vv}{{\bf v}}

\newcommand{\scarrafone}{\mathscr{P}}
\newcommand{\w}{{\bf w}}
\newcommand{\e}{\mathbf{e}}

\newcommand{\E}{\mathcal{E}}
\newcommand{\PSD}{\mathrm{S^+}}

\renewcommand{\P}{\Pi_{\mathcal P}}
\newcommand{\Q}{{\Pi_{\mathcal{Q}}}}

\newcommand{\B}{\mathcal{B}}
\renewcommand{\H}{\mathcal{H}}
\newcommand{\BH}{{\mathcal B(\mathcal H)}}
\newcommand{\D}{D}
\renewcommand{\S}{\mathcal{S}}
\newcommand{\I}{\mathcal{I}}

\newcommand{\W}{\mathcal{W}}
\newcommand{\K}{\mathcal{K}}

\newcommand{\Y}{\mathcal{Y}}

\newcommand{\id}{\mathcal{I}}

\newcommand{\F}{\mathcal{F}}
\newcommand{\V}{\mathcal{V}}
\newcommand{\CP}{\mathrm{CP}}
\newcommand{\cone}{\mathrm{cone}}
\newcommand{\sa}{\mathrm{sa}}

\newcommand*{\fracc}[2]{% \newfaktor{#1}{#2} -> #1/#2
  \raisebox{0.33\height}{\ensuremath{#1}}% Numerator
  \mkern-5mu\diagup\mkern-4mu% Slash /
  \raisebox{-0.33\height}{\ensuremath{#2}}% Denominator
}

\newtheorem{theorem}{Theorem}
\newtheorem{lemma}[theorem]{Lemma}
\newtheorem{proposition}[theorem]{Proposition}

\newtheorem{definition}[theorem]{Definition}

\newtheorem{example}{Example}
\title{Quantum learning of classical stochastic processes: \protect\\ 
       The Completely-Positive Realization Problem}
\author[1]{Alex Monr\`as}
\author[1,2]{Andreas Winter}
\affil[1]{F\'isica Te\`orica: Informaci\'o i Fen\`omens Qu\`antics,\protect\\
Universitat Aut\`onoma de Barcelona, 08193 Bellaterra (Barcelona)}
\affil[2]{ICREA -- Instituci\'o Catalana de Recerca i Estudis Avan\c{c}ats,\protect\\ Pg.~Lluis Companys, 23, 08010 Barcelona, Spain}

\date{}

\begin{document}
\maketitle

\thispagestyle{empty}

\begin{abstract}
Among several tasks in Machine Learning, a specially important one
is the problem of inferring the latent variables of a system and their 
causal relations with the observed behavior. 
A paradigmatic instance of this is the task of inferring the 
Hidden Markov Model underlying a given stochastic process. 
This is known as the positive realization problem (PRP)~\cite{benvenuti_tutorial_2004},
and constitutes a central problem in machine learning. 
The PRP and its solutions have far-reaching consequences in many areas of 
systems and control theory, and is nowadays an important piece in the 
broad field of positive systems theory~\cite{luenberger_positive_1979}.

We consider the scenario where the latent variables are quantum 
(\emph{i.e.} quantum states of a finite-dimensional system), and the 
system dynamics is constrained only by physical transformations on the 
quantum system. The observable dynamics is then described by a quantum 
instrument, and the task is to determine which quantum instrument --if any--
yields the process at hand by iterative application.

We take as a starting point the theory of quasi-realizations, whence a 
description of the dynamics of the process is given in terms of linear 
maps on state vectors and probabilities are given by linear functionals 
on the state vectors. This description, despite its remarkable 
resemblance with the Hidden Markov Model, or the iterated quantum 
instrument, is however devoid of any stochastic or quantum mechanical 
interpretation, as said maps fail to satisfy any positivity conditions. 
The Completely-Positive realization problem then consists in determining 
whether an equivalent quantum mechanical description of the same process exists.

We generalize some key results of stochastic realization theory, and 
show that the problem has deep connections with operator systems theory, 
giving possible insight to the lifting problem in quotient operator 
systems. Our results have potential applications in quantum machine 
learning, device-independent characterization and reverse-engineering of 
stochastic processes and quantum processors, and more generally, of dynamical 
processes with quantum memory~\cite{guta_fisher_2011, guta_systems_2013}.
\end{abstract}

%%%%%%%%%%%%%%%%%%%%%%%%%%%%%%%%%%%%%%%%%%%%%
%%%%%%%%%%%%%%%%%%%%%%%%%%%%%%%%%%%%%%%%%%%%%
\section{Introduction}
\label{sec:intro}
Let $\M$ be an alphabet with $|\M|=m$ symbols and let $\M^\ell$ be the set of words of length $\ell$. Let $\M^*$ be the free monoid generated by $\M$, \emph{i.e.}, the set of all finite words
\begin{equation}
	\M^*=\bigcup_{\ell\geq0} \M^\ell,
\end{equation}
with concatenation as ``multiplication'' and the empty word $\emptyword$ as unit element.
We will be concerned with stochastic processes defined on sequences of random variables over $\M$, i.e., stationary probability measures over $\M^*$. We assume throughout that $p$ is a stationary stochastic process on $\M^\infty$, namely, 
\begin{equation}
	p(\uu)\equiv p(\mathcal Y_t=u_1,\mathcal Y_{t+1}=u_2,\ldots,\mathcal Y_{t+\ell-1}=u_\ell),\quad	\uu=(u_1,\ldots,u_\ell)\in\M^\ell
\end{equation}
is independent of $t$. We will use $\ell$ to denote a generic length of a word $\uu$, so that $\uu$ can be written as $\uu=(u_1,\ldots,u_\ell)$ instead of the more cumbersome $\uu=(u_1,\ldots,u_{|\uu|})$. Let $p$ be a stationary stochastic process defined on the alphabet $\M$. Our focus will be on processes generated by ``Markovian'' dynamics of a ``finite memory''; much of the subsequent developments is to make these notions precise.

\begin{definition} 
A quasi-realization of a stochastic process is a quadruple $(\V, \pi,$ $D,\tau)$,
where $\V$ is a vector space, $\tau\in\V$, $\pi\in\V^*$ and 
$D:\M^*\rightarrow \mathcal L(\V)$ is a unital representation of 
$\M^*$ over $\V$, \emph{i.e.},
\begin{align}
	D^{(\emptyword)}&=\mathbb{1},\\
\label{eq:Drepresentation}
	D^{(\uu)}D^{(\vv)}&=D^{(\uu\vv)},\quad \forall \uu,\vv\in \M^*.
\end{align}
In addition, the following relations hold,
\begin{equation}
\label{eq:eigenrelations}
	\pi \left[\sum_{u\in\M}D^{(u)}\right] = \pi,
	\qquad
	\left[\sum_{u\in\M}D^{(u)}\right]\tau=\tau,
\end{equation}
and
\begin{equation}\label{quasi-realization}
	p(\uu)=\pi^\top D^{(\uu)}\tau\quad \forall \uu\in\M^*.
\end{equation}
\end{definition}

As can be readily seen from the definition, the quasi-realization of a 
stochastic process is far from unique. 

\begin{definition}
Two quasi-realizations are said to be {\bf equivalent} 
if they generate the same stochastic process.
\end{definition}

\medskip
The problem of deciding whether two quasi-realizations are equivalent is 
known as the \emph{identifiability} problem. Let us now just point out 
the most immediate case of equivalence. Given a quasi-realization 
$(\V, \pi,D,\tau)$, an equivalent quasi-realization can be obtained by a 
similarity transformation $T$, $(\V, \pi T,$ $T^{-1}DT,T^{-1}\tau)$. 
Such equivalent quasi-realizations are said to be {\bf isomorphic}; note 
that not all equivalent realizations are isomorphic.

Any smallest dimensional quasi-realization admitted by $p$ is called a 
{\em regular} realization, and its dimension is the {\em order} of $p$. 
Given the finite description consisting of the matrices $\{D^{(u)} : u\in\M\}$, 
and of $\tau\in \V$, $\pi\in\V^*$, the $p(\uu)$ are efficiently computable
for all $\uu=(u_1,\ldots,u_\ell)\in\M^*$ because of Eqs.~\eqref{eq:Drepresentation} 
and \eqref{quasi-realization}. Conversely, the regular realization is 
efficiently computable given the probabilities of words of length $2r-1$, 
where $r$ is the order of $p$~\cite{erickson_functions_1970,vidyasagar_complete_2011}. 
We will always consider finite dimensional spaces, thus, when we say regular 
quasi-realization we imply finite-dimensionality.

\begin{example}
\label{ex:quasirealization} 
Consider $\V=\mathbb{R}^4$ and $\M^*$ generated by the alphabet 
\[
	\M=\{+,-,x,y,z,t\}.
\]
A quasi-realization is then given by
\begin{align}
	D^{(+)}&=\frac{\gamma}{2}\left(
		\begin{array}{cccc}
 			\frac{1}{2} & 0 & 0 & \frac{1}{2} \\
 			0 & 0 & 0 & 0 \\
 			0 & 0 & 0 & 0 \\
 			\frac{1}{2} & 0 & 0 & \frac{1}{2} \\
		\end{array}
	\right), 
	&D^{(-)}&=\frac{\gamma}{2}\left(
		\begin{array}{cccc}
			\frac{1}{2} & 0 & 0 & -\frac{1}{2} \\
			0 & 0 & 0 & 0 \\
			0 & 0 & 0 & 0 \\
			 -\frac{1}{2} & 0 & 0 & \frac{1}{2} \\
		\end{array}
	\right)\\
	D^{(x)}&=\frac{\gamma}{6}\left(\begin{array}{cccc}
		1&0&0&0\\
		0&1&0&0\\
		0&0&\cos\theta&\sin\theta\\
		0&0&-\sin\theta&\cos\theta
	\end{array}\right), 	
	&D^{(y)}&=\frac{\gamma}{6}\left(\begin{array}{cccc}
		1&0&0&0\\
		0&\cos\theta&0&-\sin\theta\\
		0&0&1&0\\
		0&\sin\theta&0&\cos\theta
	\end{array}\right)\\
	D^{(z)}&=\frac{\gamma}{6}\left(\begin{array}{cccc}
		1&0&0&0\\
		0&\cos\theta&\sin\theta&0\\
		0&-\sin\theta&\cos\theta&0\\
		0&0&0&1
	\end{array}\right), 	
	&D^{(t)}&=(\gamma-1)\left(\begin{array}{cccc}
		-1&0&0&0\\
		0&1&0&0\\
		0&0&1&0\\
		0&0&0&1
	\end{array}\right)
\end{align}
and vectors/covectors
\[
	\pi=(1,0,0,0),\quad
	\tau=\left(\begin{array}{c}1\\0\\0\\0\end{array}\right).
\]
One can check that the eigenvector conditions are satisfied for 
this set of matrices and pair of vectors.
\end{example}

That a quasi-realization describes a normalized and stationary measure 
on $\M^*$ follows easily from the eigenvector conditions \eqref{eq:eigenrelations}. 
However, whether the measure is nonnegative is a nontrivial question.

\begin{lemma}
\label{nonnegative}
A quasi-realization defines a nonnegative measure if and only if there 
is a convex cone $\C\subset\V$ such that
\begin{enumerate}
	\item $\tau\in\C$,
	\item $D^{(\uu)}(\C)\subseteq\C$, and
	\item $\pi\in\C^* = \{ f\in\V^* : f(x) \geq 0 \ \forall x\in\C \}$.
\end{enumerate}
\end{lemma}
\begin{proof}
It is easy to see that $p(\uu)\geq0$ follows from the conditions above. 
The converse follows by considering $\C=\cone\{p(\uu)^{-1}D^{(\uu)}\tau:\uu\in\M^*\}$ 
as the conic hull of all the vectors $D^{(\uu)}\tau$. Then $\tau\in\C$ and 
$D^{(\uu)}(\C)\subseteq\C$ by construction. That $\pi\in\C^*$ follows from 
$p(\uu)\geq0$.
%\hfill$\blacksquare$
\end{proof}

\medskip
The general problem of deciding whether a quasi-realization defines a 
nonnegative measure is equivalent to deciding whether a rational power 
series always has nonnegative coefficients, a problem known to be 
undecidable~\cite{sontag_questions_1975}. A generalization of this 
undecidability result has been recently shown for matrix-product 
operators~\cite{kliesch_matrix_2014}. Below, we will show that the 
above example induces a nonnegative measure by explicitly providing 
a cone $\C$ (see Examples~\ref{ex:noclassical}, \ref{ex:SDR1},
\ref{ex:twoqubits} and \ref{ex:qubit+qubit}). 
Before, we derive some general properties about cones associated to 
regular realizations.

\begin{lemma}
\label{LEM:ORDERUNIT}
Let $\mathcal R=(\V,\pi,D^{(\uu)},\tau)$ be a regular quasi-realization, \emph{i.e.}, no equivalent quasi-realization exists of smaller dimension. Then, any cone $\C$ satisfying conditions in Lemma~\ref{nonnegative} is \emph{proper}: $\C$ does not contain nor is contained in a proper subspace of $\V$, and $\tau$ is an order unit of $\C$.
\end{lemma}
The proof is given in Appendix~\ref{app:proof1}.

\bigskip
The rest of the paper is organized as follows: 
In the next Section~\ref{sec:HMM} we review Hidden Markov Models, quasi-realizations,
the associated positive realization problem (PRP) and a characterization in terms
of stable polyhedral cones. This puts into the focus generalized probabilistic
theories (GPTs), which provide quasi-state spaces to explain processes.
In Section~\ref{sec:QMC} we move on to the object of our study, Quantum Hidden
Markov Models, and the completely positive realization problem (CPRP).
The subsequent sections are devoted to developing the theory leading to our
characterization of the CPRP in terms of semidefinite representable (SDR) cones:
In Section~\ref{sec:quotients} we introduce some general results about identifiability which introduce the notion of quotient realization, which will be central throughout the paper.
In Section~\ref{sec:SDR+quotients} we pick up the notion of quotient realization and show how this naturally determines the class of cones that will be of our interest, namely, semidefinite representable cones. Then, complete positivity leads us to consider quotient operator systems, and show how semidefinite representable cones arise naturally in this context, as the positive sets in $\V\otimes \B(\mathbb{C}^n)$.
Then, in Section~\ref{sec:regular-is-quotient}, we use the fact that every regular quasi-realization is a quotient realization to derive necessary conditions for the feasibility of the CPRP. We show how in certain specific cases, these conditions can be also sufficient, and discuss why in general this is not the case. Then, we present the main subject of our result, namely, the SDR mapping cone.
Section~\ref{sec:ogni-scarraffon-e-bbello-a-mmamma-suia} finally presents and
proves our main result: That the existence of a compatible SDR mapping cone is necessary and sufficient for the feasibility of the CPRP.
Throughout, we develop the simple Example~\ref{ex:quasirealization} above,
exhibiting its nonclassical and indeed surprising quantum nature as we go
along, to illustrate the theoretical concepts.

%%%%%%%%%%%%%%%%%%%%%%%%%%%%%%%%%%%%%%%%%%%%%
%%%%%%%%%%%%%%%%%%%%%%%%%%%%%%%%%%%%%%%%%%%%%
\section{Classical learning problem: Hidden Markov Models}
\label{sec:HMM}
A central task in machine learning is to obtain the latent 
variables that account for the apparent complexity of a given 
process $p$. These variables, although not directly accessible 
to the observable dynamics summarize past behavior while still 
providing complete information about future probabilities of 
events. To accomplish this, one aims to find a random variable 
$X$ such that the future is independent of the past, given $X$:
\begin{equation}
  \label{eq:factor}
  p(\vv|\uu)=\sum_{X}P(\vv|X)P(X|\uu).
\end{equation}
However, for such a decomposition to exist at any given time we 
require that state transition probabilities are only dependent 
on the generated output, 
\begin{equation}
  P(X_t=x', \Y_t=u|X_{t-1}=x) \equiv [M^{(u)}]_{x'x},
\end{equation} 
in a time-invariant manner. This implies that $X$ is Markovian, 
and we say that $p$ is a Hidden Markov Process. In such case, 
$\{X_t\}$ represents the latent variables of $p$, and an important 
problem in machine learning consists in recovering the transition 
probabilities $M^{(u)}$.

A processes quasi-realization constitutes an abstract model of the behavior of $p$. 
However this does not suffice to identify its latent variables in an operational sense.
Indeed, the elements of a quasi-realization $(\V, \pi,D,\tau)$ are just a vector
space and some elements, maps and dual elements, with a priori no meaning attached 
to their coordinates.

In particular, the vector $\pi$ does not necessarily satisfy any positivity criterion, 
and the maps $D^{(\uu)}$ need not be related to any stochastic transition probabilities. 
Moreover, the vectors $\pi D^{(\uu)}$ will potentially take an unbounded number of distinct 
values over $\V$, giving little insight into the essential mechanism driving~$p$.

\begin{definition}
A \textbf{\emph{positive} realization} of $p$ is a quasi-realization $(\mathbb{R}^d,\pi,M,\tau)$, such that $M^{(\uu)}$ are nonnegative matrices and
\begin{align}
	M=\sum_{u\in\M}M^{(u)} \text{ is stochastic},
\end{align}
$\pi\in (\mathbb{R}^d)^*$ is a stationary distribution of $M$, and $\tau=(1,1,\ldots,1)\in \mathbb{R}^d$.
\end{definition}

Note that in a positive realization, the cone required in Lemma~\ref{nonnegative} 
is simply $\C=\mathbb{R}^d_+$, the cone of entry-wise nonnegative vectors.

\medskip
\noindent{\bf The Positive Realization Problem} (PRP) is the problem of finding a positive realization of a process $p$, given its regular realization~\cite{vidyasagar_complete_2011}. This problem stands as one of the main quests in machine learning, and several variants can be considered depending on the application. It remains open in its more general formulation, although certain general aspects of it are well-understood~\cite{benvenuti_tutorial_2004}. It has been addressed from a variety of perspectives~\cite{anderson_realization_1999,vidyasagar_complete_2011,cybenko_learning_2011}. However, a full solution remains open and its computational cost is still unknown.  One of the keystones of this problem, was established early on.

\begin{theorem}[Positive realization problem~\cite{dharmadhikari_sufficient_1963}]
\label{thm:cone} 
Given a quasi-realization $(\V,\pi,D,\tau)$, an equivalent positive realization exists if and only if there is a convex pointed \textbf{polyhedral} cone $\C\subset \V$ such that
\begin{enumerate}
\item $\tau\in \C$
\item $D^{(v)}(\C)\subseteq \C$,
\item $\pi\in\C^*$
\end{enumerate}
where $\C^*$ is the dual cone of $\C$.
\hfill\qedsymbol
\end{theorem}

This theorem highlights the nature of the cone $\C$ that guarantees not only positivity of the induced measure $p$ on $\M^*$, but also on the level of the vector $\pi$ and maps $D^{(\uu)}$. A polyhedral condition on the cone $\C$ gives a nice geometrical interpretation of the mechanisms underlying the system's dynamics, where each of the generating vectors of $\C$ is mapped onto a nonnegative linear combination of all the vertices. This insight allows to draw equivalent positive realizations from nonpositive ones~\cite{vidyasagar_complete_2011}. This, however, is far from being a solution to the problem. The difficulty in pushed into finding the polyhedral cone that satisfies the conditions~\cite{van_den_hof_positive_1999}. Partial results in this direction have been obtained~\cite{anderson_realization_1999, vidyasagar_complete_2011}. 

\begin{example}
\label{ex:noclassical} 
We revisit the quasi-realization of Example~\ref{ex:quasirealization} to note that in 
the event of $\theta/2\pi$ being irrational, elements of the monoid generated by 
$\{x,y,z\}$ will be arbitrarily close to any element of $1\oplus SO(3)$. This implies 
that a polyhedral stable cone $\C$ cannot exist, as is seen by noticing that any 
such cone must be stable under all $1\oplus SO(3)$ transformations.
\end{example}

This can be stated as a more general result.

\begin{lemma} 
Given a quasi-realization $(\V,\pi,D^{(\uu)},\tau)$, 
if there is a group $\mathcal G$ such that for all $g\in \mathcal{G}$
there exist an element $\uu_g\in\M^*$ and a scalar $c_g > 0$ such that
\begin{align}
  c_g c_{g'} D^{(\uu_g)}D^{(\uu_{g'})} = c_{gg'}D^{(\uu_{gg'})},
\end{align}
then any cone $\C$ satisfying the conditions in Lemma~\ref{nonnegative} 
must be invariant under the representation $g\rightarrow c_g D^{(\uu_g)}$.
\end{lemma}

\begin{proof}
It is easy to show that for all $\{\uu_g,g\in \mathcal G\}$, the stability condition becomes invariance. Indeed, combining $D^{(\uu_{g^{-1}})}(\C)\subseteq\C$ with $D^{(\uu_g)}(\C)\subseteq\C$ we have
\[
	D^{(\uu_g)}D^{(\uu_{g^{-1}})}(\C)\subseteq D^{(\uu_g)}(\C),
\]
so that $c_{e}/(c_gc_{g'})\C\subseteq D^{(\uu_g)}(\C)$, where $e$ is the identity of $\mathcal G$. Since a cone is invariant under multiplication by a positive constant, inclusion in the opposite sense follows, and thus equality $D^{(\uu_g)}(\C)=\C, \forall g\in\mathcal G$.
\end{proof}

A simple corollary is that if the realization generates a nontrivial 
representation of a continuous group (in the sense that it is not constant 
along all connected components), then  the stable cone must be invariant 
under such continuous symmetry. This rules out all polyhedral cones and 
thus the existence of equivalent positive realizations.

It is worth stressing the prominent role of polyhedral geometry in the 
constructions of stochastic models, a feature that has been revisited 
repeatedly in the study of quantum {\em vs} classical 
correlations~\cite{werner_bell_2001} and from  the perspective of more 
general probabilistic theories~\cite{pfister_one_2012}.
These latter are a theoretical framework that originated in the foundations of 
quantum mechanics \cite{Ludwig:1,Ludwig:2,Holevo:stat-struc-book,Barrett-et-al:PRA},
motivated by the realization that quantum mechanical state space (and also
the space of measurements) are merely there to predict probabilities of 
observations -- but these spaces are cones of semidefinite matrices rather than
simplicial cones characteristic of probability theory.
It has been shown that these ``generalized'' state spaces allow for the
development of rudimentary physical phenomenology, involving the statistical
elements of quantum mechanics --preparation, measurement and state transformation--,
serving as toy models or blueprints for quantum mechanical behaviour. The
implicit message is that, rather than being a purely mathematical construct,
the cone implied by a quasi-realization (Lemma~\ref{nonnegative}) describes
a potentially ``real'', if exotic, system. This is a fruitful point of view
not only for foundational investigations (where one wants to compare the
probability rule of quantum mechanics with alternatives), but also in our
present context. Note indeed that Dharmadhikari's Theorem~\ref{thm:cone}
above shows that polyhedral cones are naturally interpreted as effective
state spaces of a positive realization, \emph{i.e.}~a classical probabilistic 
model of the observed process in disguise.

%%%%%%%%%%%%%%%%%%%%%%%%%%%%%%%%%%%%%%%%%%%%%
%%%%%%%%%%%%%%%%%%%%%%%%%%%%%%%%%%%%%%%%%%%%%
\section{Quantum learning problem: Finitely correlated processes}
\label{sec:QMC}

We address the natural quantum generalization of the PRP, namely, when the relevant information about the past can be synthesized by a quantum state, rather than a classical random variable. This requirement, less restrictive than the classical one~\cite{monras_hidden_2011}, has been considered from the perspective of $\epsilon$-machines~\cite{gu_quantum_2012}, where it was shown that the statistical complexity of the system could be reduced by a quantum model. Instead, our approach focuses on the dimension of the quantum system, which can be drastically reduced once one allows for quantum states. A highly relevant example in a not too distant scenario can be found in~\cite{wolf_assessing_2009}. 

In the quantum mechanical setting, the factorization condition Eq.~\eqref{eq:factor} is replaced by
\begin{equation}
	p(\vv|\uu)=\rho_\uu[ E^{(\vv)}],
\end{equation}
where $\rho_\uu$ represents a quantum state associated with history $\uu$, and $E^{(\vv)}\geq0$ is the POVM element associated with future outcome $\vv$. Future probabilities are obtained by the Born rule applied to state $\rho_\uu$. The minimum dimension by which this description can be achieved is given by the positive semidefinite rank~\cite{fiorini_linear_2012}. However, in addition, in order to have a physically meaningful description of the mechanisms at work, one expects that the state transitions are given by physical transformations, 
\begin{equation}
	\rho_{\uu v}=\rho_\uu\circ \E^{(v)},
\end{equation}
where $\E^{(v)}$ are completely positive maps, and $\sum_{v\in\M}\E^{(v)}$ is unital. The set $\{\E^{(v)}\}$ is called a {\em quantum instrument}, and describes all relevant properties of a generalized quantum measurement, be it for the purpose of computing outcome probabilities by providing the corresponding POVM elements, $E^{(v)}=\E^{(v)}(\id)$, as well as for describing the conditional post-measurement state $\rho_v=\E^{(v)}{}^*(\rho)$.

\begin{definition} 
A \emph{\textbf{completely positive realization}} is a quasi-realization 
$(\BH^\sa,$ $\rho, \E, \id)$, where $\BH$ is the space of bounded operators on 
some finite-dimensional Hilbert space $\H$ and $\BH^\sa$ the (real) subspace of selfadjoint
operators, $\rho$ is a positive semidefinite density operator in $\BH$, 
$\E:\M^*\rightarrow \B(\BH)$ is a unital completely positive map-valued 
representation of the free monoid $\M^*$ and $\id$ is the identity of $\BH$.
\end{definition}
%
%A completely positive realization allows to understand the system's state $\rho\circ \E^{(\uu)}$ as an unnormalized conditional quantum state
%\begin{align}
%	\rho_\uu=\frac{\E^{(\uu)}{}^*(\rho)}{p(\uu)},
%\end{align}
%which provides the future outcome probabilities $p(\vv|\uu)$ by the Born rule,
%\begin{align}
%	p(\vv|\uu)&=\rho_\uu(\E^{(\vv)}(\id))\\
%		&=\frac{\rho\circ\E^{(\uu\vv)}(\id)}{p(\uu)}.
%\end{align}
%This illustrates that $E^{(\vv)}=\E^{(\vv)}(\id)$ is the POVM associated to outcome $\vv$, which combined with all other strings of length $|\vv|=\ell$ constitutes a properly defined POVM, $\sum_{\vv\in\M^\ell}E^{(\vv)}=\id$. Furthermore, the state transitions associated to a given outcome $\vv$ are given in terms of physical nondeterministic maps.

\begin{example} 
\label{ex:invertiblequbit}
A quasi-realization equivalent, in fact isomorphic,
to that of Example \ref{ex:quasirealization} 
is given by $\V = \B(\mathbb{C}^2)^\sa$, with the maps
\begin{align}
	\E^{(\pm)}(X) &= \frac{\gamma}{2}P_\pm X P_\pm,\\
	\E^{(j)}(X) &= \frac{\gamma}{6}e^{i\theta/2\sigma_j}X 
	                                 e^{-i\theta/2\sigma_j},\quad j\in\{1,2,3\},\\
     \label{eq:transpose}
	\E^{(t)}(X)   &= (1-\gamma)\sigma_2 X^\top \sigma_2,
\end{align}
where $\sigma_\mu$ are the $2\times 2$ Pauli matrices and 
$P_\pm=\frac{\id\pm \sigma_3}{2}$ are the up/down spin projections. 
The stationary state is $\pi=\frac12\id$ (strictly speaking, it is the 
functional $\frac12\tr$, \emph{i.e.} the normalized trace), the 
stationary element $\tau = \I$, the matrix identity.

The stable convex cone required by Lemma~\ref{nonnegative} is precisely the
set $\C = \V^+ = \B(\mathbb{C})^+$ of positive semidefinite matrices,
which is indeed left invariant by all maps $\E^{(u)}$, contains
the operator $\I$ and the dual cone $\C^*$, which is isomorphic
to $\B(\mathbb{C})^+$, contains $\pi$.

This quasi-realization is completely positive when $\gamma=1$, but 
for $\gamma<1$ the map $\E^{(t)}$ is not completely positive. The 
physical interpretation for this realization is the following: 
At every time step, the system undergoes a projective measurement 
in the computational basis with probability $1/2$, and with probability 
$1/6$ undergoes a $\theta$-rotation around a random axis $\mu\in(x,y,z)$. 
As discussed in Example \ref{ex:noclassical}, when $\theta/2\pi$ 
is irrational, the model has no equivalent positive realization. 
This can be understood as the system's qubit potentially occupying 
any point in the Bloch sphere, thus no classical means of encoding 
the quantum information of $\rho$ will suffice, unless it use infinite
memory.
\end{example}

\medskip\noindent 
{\bf The completely positive realization problem (CPRP):} 
\emph{Given a quasi-realization of process $p(\uu)$, determine whether there is an equivalent completely positive realization, \emph{i.e.}, find a quantum instrument $\{\E^{(u)}\}$, and positive semidefinite $\rho$ such that 
\begin{equation}
	p(\uu)=\rho[\E^{(u_1)}\circ \cdots\circ\E^{(u_\ell)}(\id)],
\end{equation}
and $\left(\sum_{u\in\M}\E^{(u)}\right)^*(\rho)=\rho$. 
Stochastic processes admitting a completely positive 
realization are called \emph{finitely correlated} or 
\emph{algebraic}~\cite{fannes_finitely_1992}, special cases being classical
Markov chains, Hidden Markov Models, and the previous notion of quantum Markov 
chains~\cite{accardi_noncommutative_1978}.}

\medskip
Compared to the large amount of effort devoted to the PRP, the CPRP has 
received significantly less attention in the literature. It arises naturally 
--albeit in slight disguise-- in~\cite{blume_kohout_robust_2013}, and more 
generally in quantum systems identification~\cite{burgarth_quantum_2012, guta_systems_2013, arenz_control_2014, burgarth_identifiability_2014}. 
It is worth mentioning that any positive realization can be cast as a 
completely positive one of the same (Hilbert space) dimension, so the CPRP 
cannot be harder than the PRP.

%%%%%%%%%%%%%%%%%%%%%%%%%%%%%%%%%%%%%%%%%%%%%
%%%%%%%%%%%%%%%%%%%%%%%%%%%%%%%%%%%%%%%%%%%%%
\section{Quotient realizations and identifiability}
\label{sec:quotients}
In a positive realization, the polyhedral nature of the cone $\C$ 
is a consequence of the classical nature of the process, and it can 
be understood, ultimately, as a consequence of the polyhedral nature 
of the simplex $\mathbb{R}^d_+$. We discuss now the precise way in 
which such features of the underlying mechanism are revealed as properties 
of the cone $\C$. 

In order to obtain necessary and sufficient conditions on a quasi-realization for the existence of an equivalent completely positive one, we first generalize a classical result by Ito, Amari and Kobayashi~\cite{ito_identifiability_1991}. The latter is the stochastic equivalent to a classic result on linear systems theory~\cite{kalman_irreducible_1965}, saying that minimal realizations are always isomorphic (\emph{i.e.}, related by similarity transformations), and are quotients of higher-dimensional ones.

Let $\mathcal{R}=(\mathcal{U}, \pi, D,\tau)$ be a quasi-realization. 
The following aims at understanding the structure of the smallest 
dimensional equivalent realization. Define  
$\W=\vecspan\{D^{(\uu)}\tau\}_{\uu\in\M^*}\subseteq\mathcal U$ as 
the \emph{accessible} subspace. It is trivially stable under the 
action of $D^{(\uu)}$ for all $\uu\in\M^*$:
\begin{equation}
	D^{(\uu)}\W\subseteq \W.
\end{equation}
Analogously, consider the span of states $\widetilde \W=\vecspan\{\pi D^{(\uu)}\}_{\uu\in\M^*}$. 
Its annihilator, $\widetilde \W^\perp =\bigcap_{\sigma\in\widetilde{\W}}\ker \sigma$, 
is the \emph{null} space, {\em i.e.}~the subspace which has no effect whatsoever for 
computing word probabilities. Also $\widetilde \W^\perp $ is stable under $D^{(\uu)}$
for all $\uu\in\M^*$:
\begin{equation}
	D^{(\uu)}\widetilde \W^\perp\subseteq \widetilde \W^\perp,
\end{equation}
since $\widetilde{\W}D^{(\uu)} \subseteq \widetilde{\W}$.

Define the \emph{quotient} space $\V$ as the accessible space modulo its 
null component $\K=\W\cap \widetilde \W^\perp$:
\begin{equation}
	\V \equiv \fracc{\W}{\K}.
\end{equation}
The elements of $\V$ are of the form $a+\K,~a\in \W$. Let $L:\W\rightarrow\V$ be the canonical projection onto $\V$,
\begin{equation}
\begin{array}{lllll}
	L:	
	&\W		&\rightarrow &\V\\
			&v	&\mapsto		&v + \K. 
\end{array}
\end{equation}
Since $D^{(\uu)}\K \subseteq \K$, we have a well-defined quotient map 
$\overline D^{(\uu)}:\V\rightarrow \V$, with the property 
$\overline D^{(\uu)}\circ L=L \circ D^{(\uu)}$. Also, define 
$\bar \tau=L(\tau)$ and $\bar \pi$ as the induced quotient functional 
$\bar \pi\circ L=\pi$. Using the fact that $\pi[\ker L]=0$ we factor 
through the entire set of maps $D^{(u)}$,
\begin{align}
	p(\uu) &= \pi\circ D^{(\uu)}(\tau)\\
		   &= \bar \pi\circ L\circ D^{(\uu)}(\tau)\\
		   &= \bar \pi\circ \overline D^{(\uu)}(\bar \tau).
\end{align}
This, together with easily shown eigenvector relations \eqref{eq:eigenrelations} proves that 
$(\V,\bar \pi, \overline D,\bar \tau)$ constitute a perfectly valid quasi-realization. 
Furthermore $(\V,\bar \pi,\overline D,\bar \tau)$ is equivalent to 
$(\mathcal{U},\pi,D,\tau)$. We call such quasi-realization the 
\emph{quotient} realization. 
An important step is to realize that equivalent quotient realizations are isomorphic.

\begin{theorem}[\cite{ito_identifiability_1991}]
\label{THM:ISOMORPHISM}
Two quasi-realizations $\mathcal{R}_1=(\V_1,\pi_1,D_1,\tau_1)$ and 
$\mathcal{R}_2=(\V_2,\pi_2,D_2,\tau_2)$ of the same stochastic process 
$p$, not necessarily of the same dimension, have isomorphic quotient 
realizations 
$\overline{\mathcal{R}}_i=(\overline\V_i,\overline{\pi}_i,\overline{D}_i,\overline{\tau}_i)$, 
$i=1,2$:
It holds $\overline\V_1\stackrel{T}\cong\overline\V_2$, and
\begin{align}
 	\overline\pi_1      &=\overline\pi_2 T,\\
	\overline D_1^{(u)} &=T^{-1}\overline D_2^{(u)}T,\\
	\overline\tau_1     &=T^{-1}\overline\tau_2.
\end{align} 
\end{theorem}

This result follows from \cite{ito_identifiability_1991}, which proves it however
only for the Hidden Markov Model case. But the proof only relies on the 
nonnegativity of the process probabilities, and applies to any pair of 
equivalent and well-defined quasi-realizations (in the sense that they yield 
the same nonnegative measure on $\M^*$). For the sake of completeness,
we reproduce the proof in Appendix~\ref{app:isomorphism} with our notation.

This result is important in that it establishes the uniqueness of the quotient space $\V$, up to basis changes. Let $d$ be the dimension of $\V$. As can be seen from the definition, 
$d=\dim \V\leq n=\dim \mathcal U$, where $n$ is the original realization's dimension. By considering the quotient of a regular realization of dimension $r$ we get $d\leq r$. On the other hand $r$ is a lower bound to the dimension of any equivalent quasi-realization. Thus we conclude that $d=r$, hence quotient realizations are indeed regular, and all regular realizations can be regarded as quotient realizations. 
The concepts underlying realization theory have recently been used in quantum systems 
identification~\cite{burgarth_quantum_2012, guta_systems_2013, arenz_control_2014, burgarth_identifiability_2014}, and a similar result for multipartite quantum 
systems appeared in~\cite{guta_equivalence_2014}.

\begin{example}
\label{ex:SDR1} 
Consider the realization given in Example~\ref{ex:invertiblequbit} regardless 
of whether it is completely positive or not. It is defined in the $4$-dimensional 
vector space $\B(\mathbb{C}^2)^\sa$. We will show that its quotient is 
precisely that of Example~\ref{ex:quasirealization}. Its accessible subspace 
is the self-adjoint subspace, $\W=\B(\mathbb{C}^2)^\sa$, spanned by the Pauli 
matrices  and the identity, 
\begin{equation}
	\W=\vecspan\{\sigma_\mu : \mu=0,1,2,3\},
\end{equation}
and analogously $\widetilde \W=(\B(\mathbb{C}^2)^\sa)^*$, so 
$\widetilde \W^\perp\cap\W=\{0\}$. Therefore $L$ is just an isomorphism 
$L:\W\rightarrow\V\cong\mathbb{R}^4$. Fixing a basis in $\V$ as 
$\e_\mu=L(\sigma_\mu)$ yields the regular realization of Example~\ref{ex:quasirealization}. 

Because all maps in Example~\ref{ex:invertiblequbit} are positive 
(although not completely positive), they satisfy $\E^{(u)}(\PSD)\subseteq\PSD$. 
This implies that $D^{(u)}L(\PSD)\subseteq L(\PSD)$, while $\tau=L(\id)$ and 
$\pi\circ L=\rho$. Therefore, the cone $\C=L(\PSD)$ is given by
\begin{align}
	\C &= \left\{x\in\V\middle| \sum_\mu x_\mu \sigma_\mu\geq0\right\}\\
	   &= \left\{x\in\V\middle| x_0\geq|\vec x|\right\},
\end{align}
which is a parameterization of the well-known positive-semidefinite cone 
$\PSD$. We have introduced a notation that will be exploited in the following, 
namely $x_\mu=(x_0,\vec x)$, and $|\vec x|=\sqrt{\sum_{i=1}^3 x_i^2}$.
\end{example}

This example illustrates how the cone $\C$ referred to in Lemma~\ref{nonnegative} can be derived from the preserved cone in its higher dimensional equivalent realizations. This equivalent realization, however, is not completely positive. We next show two different completely positive realizations that, despite not being isomorphic, have that of Example~\ref{ex:quasirealization} as quotient realization.

\begin{example} 
\label{ex:twoqubits}
We now construct a completely positive realization on two qubits. 
Let $\H=\mathbb{C}^2\otimes\mathbb{C}^2$, and consider the following maps in $\BH$.
\begin{align}
	\E^{(\pm)}(X) &= \frac{\gamma}{2}P_\pm\tr_B[X]P_\pm\otimes P_\mp,\\
	\E^{(j)}(X)   &= \frac{\gamma}{6}e^{i\theta/2(\sigma_j\otimes\mathbb{1}+\mathbb{1}\otimes\sigma_j)}X e^{-i\theta/2(\sigma_j\otimes\mathbb{1}+\mathbb{1}\otimes\sigma_j)},\quad j\in\{1,2,3\},\\
	\E^{(t)}(X)   &= (1-\gamma)\Sigma X \Sigma^\dagger,
\end{align}
where $\Sigma=\tfrac12\sum_\mu \sigma_\mu\otimes\sigma_\mu$ is the SWAP operator, and $\tr_B[X]$ is the partial trace on the second subsystem. One can check that these maps are completely positive for all $\gamma\in[0,1]$. In addition, we define the reference elements $\id=\mathbb{1}\otimes\mathbb{1}$, and $\rho$ is the only vector such that $\sum_{u\in\M}{\E^{(u)}}^*(\rho)=\rho$.

Excluding critical values of $\theta$ and $\gamma$ such as $\gamma=0$ or $\theta=\pi$, the observable and accessible subspaces $\W$ and $\widetilde \W$ are, respectively,
\begin{align}
  \label{eq:Wspace}
  \W &=            \vecspan\{\sigma_\mu\otimes\mathbb{1},\mathbb{1}\otimes\sigma_\mu\}_{\mu=0,\ldots,4},\\
\nonumber
  \widetilde \W &= \vecspan\{\sigma_\mu\otimes\sigma_\mu,\sigma_\mu\otimes\mathbb{1}-\mathbb{1}\otimes\sigma_\mu\}_{\mu=0,\ldots,4}\\
	            &\quad
	             + \vecspan\{\sigma_1\otimes\sigma_2+\sigma_2\otimes\sigma_1,\sigma_2\otimes\sigma_3+\sigma_3\otimes\sigma_2,\sigma_3\otimes\sigma_1+\sigma_1\otimes\sigma_3\}.
\end{align}
We have used the Hilbert-Schmidt inner product to identify elements $f\in\BH^*$ with elements in $F\in\BH$, such that $f(X)=\tr[F^\dagger X]$. The respective dimensions are $\dim\W=7$ and $\dim\widetilde\W=10$. The annihilator of $\widetilde \W$, is 
\begin{align}
\nonumber
	\widetilde\W^\perp=&\,\vecspan\{\mathbb{1}\otimes\sigma_i+\sigma_i\otimes\mathbb{1}\}_{i=1,2,3}\\
	&\,+\vecspan\{\sigma_1\otimes\sigma_2-\sigma_2\otimes\sigma_1,\sigma_2\otimes\sigma_3-\sigma_3\otimes\sigma_2,\sigma_3\otimes\sigma_1-\sigma_1\otimes\sigma_3\},
\end{align}
so that for $\K=\W\cap\widetilde\W^\perp$ we have 
\begin{align}
	\K=\vecspan\{\mathbb{1}\otimes\sigma_i+\sigma_i\otimes\mathbb{1}\}_{i=1,2,3}.
\end{align}
Clearly, $\K\cap\PSD=\{0\}$. $\K$ is 3-dimensional, so $\dim\W/\K=4$. 

We can use the customary notation $\sigma_A=\sigma\otimes\mathbb{1}$, $\sigma_B=\mathbb{1}\otimes\sigma$, and $\vec n\cdot \vec \sigma\equiv\sum_in_i \sigma_i$ to characterize elements $\omega\in\W$ by
\begin{align}\label{eq:arbitraryomega}
	\omega=&\,c_0 \mathbb{1}+ \vec n_A \cdot \vec \sigma_A+\vec n_B\cdot \vec \sigma_B.%\\
%	=&\,(c_0-|\vec n_A|-|\vec n_B|)\mathbb{1}\otimes\mathbb{1}+2|\vec n_A|\ket{\hat n_A}\bra{\hat n_A}\otimes\mathbb{1}+2|\vec n_B|\mathbb{1}\otimes\ket{\hat n_B}\bra{\hat n_B},
\end{align}
Hence, defining basis elements $\omega_\mu$ as
\begin{align}
	\omega_0=\mathbb{1},\qquad
	\omega_i=\frac12(\sigma_A{}_i-\sigma_B{}_i),
\end{align}
%where 
%\begin{align}
%	\ket{\hat n}\bra{\hat n}=\frac{\mathbb{1}+\hat n\cdot\vec \sigma}{2}.
%\end{align}
the quotient map $L:\W\rightarrow\W/\K\cong\V$ is $L(\omega_\mu)=\e_\mu$. Hence,  an arbitrary element $x\in\W/\K$ of the form $x=x_\mu\e_\mu$ corresponds to elements $\omega=x_0\mathbb{1}+\tfrac12\vec x \cdot (\vec \sigma_A-\vec \sigma_B)+k$, where $k\in \K$. The action of $\E^{(u)}$ on such elements is given by
\begin{align}
	\E^{(u)}(\omega_\mu)=\sum_{\nu}D^{(u)}_{\mu\nu}\omega_\nu+k_\mu^{(u)}
\end{align}
where $k_\mu^{(u)}\in \K$. The coefficients $D^{(u)}_{\mu\nu}$ are given in 
Example~\ref{ex:quasirealization}, and constitute the quotient realization,
\begin{align}
	L\circ \E^{(u)}=D^{(u)}\circ L.
\end{align}
\end{example}
This shows the first example of how a quasi-realization arises nontrivially from a higher dimensional completely positive realization. The fact that the maps in Example~\ref{ex:invertiblequbit} are positive already guarantees that a construction of this type (which essentially represents the transpose map as the induced quotient map of the particle exchange map, which is completely positive) will exists for dimension 2, as all such positive maps are a convex sum of completely positive and completely co-positive ones.

%%%%%%%%%%%%%%%%%%%%%%%%%%%%%%%%%%%%%%%%%%
\begin{example} 
\label{ex:qubit+qubit}
We now construct a completely positive realization on two qubits. 
For the usual matrix transpose $X^\top$ define the \emph{universal spin flip}
map $\Phi(X) \equiv \sigma_2 X^\top \sigma_2$. 
Now let $\H=\mathbb{C}^2\oplus\mathbb{C}^2$, and consider the subspace
\begin{equation}
  \label{eq:W-space}
  \W = \{ Y \oplus Z\ :\ Z=\Phi(Y),\ Y \in \B(\mathbb{C}^2)^\sa \} \subset \BH^\sa.
\end{equation}
The following maps on $\BH$, whose elements we write as $2\times 2$-block 
matrices $X=\sum_{i,j=0}^{1} X_{ij} \otimes \ketbra{i}{j}$,
can be checked to map $\W$ to itself:
\begin{align}
	\E^{(\pm)}(X) &= \frac{\gamma}{2} (P_\pm \oplus P_\mp) X (P_\pm \oplus P_\mp), \\
	\E^{(j)}(X)   &= \frac{\gamma}{6} 
	                    \left( e^{i\theta/2\,\sigma_j} \oplus e^{-i\theta/2\,\sigma_j}\right) 
	                                     X 
	                    \left( e^{-i\theta/2\,\sigma_j} \oplus e^{i\theta/2\,\sigma_j}\right), 
	                                                                  \quad j\in\{1,2,3\}, \\
	\E^{(t)}(X)   &= (1-\gamma) (\mathbb{1}\otimes\sigma_1) X (\mathbb{1}\otimes\sigma_1)^\dagger\\
	              &= (1-\gamma) \sum_{i,j=0}^{1} X_{ij} \otimes \ketbra{1-i}{1-j},
\end{align}
where we observe that $\Phi(P_\pm) = P_\mp$ and 
$\Phi(\sigma_j)=-\sigma_j$ for all $j=1,2,3$, while $\Phi(\sigma_0)=\sigma_0$.

One can check that these maps are completely positive for all $\gamma\in[0,1]$. 
In addition, we define the reference elements 
$\id=\mathbb{1}_4=\mathbb{1}_2\oplus\mathbb{1}_2$, 
and $\pi = \frac14\tr$, corresponding to the maximally mixed state.
Using the identification
\[
  \W \ni Y \oplus \Phi(Y) \leftrightarrow Y \in \V=\B(\mathbb{C}^2)
\]
one can also check that this realization, restricted to $\W$ is 
isomorphic to the one from Example~\ref{ex:invertiblequbit}, 
and hence to Example~\ref{ex:quasirealization}.
Note that the isomorphism only determines the action of the maps
$\E^{(u)}$ on $\W$, but this does not uniquely determine them. For instance,
we could have used all or any subset of the following instead:
\begin{align*}
	\widetilde\E^{(\pm)}(X) &= \frac{\gamma}{2} \bigl( (P_\pm \oplus 0) X (P_\pm \oplus 0) 
	                                             + (0 \oplus P_\mp) X (0 \oplus P_\mp) \bigr), \\
	\widetilde\E^{(j)}(X)   &= \frac{\gamma}{6} \left(
	                                                  \left( e^{i\theta/2\,\sigma_j} \!\oplus\! 0 \right) 
	                                                   X 
	                                                  \left( e^{-i\theta/2\,\sigma_j} \!\oplus\! 0 \right)
	                          + 	                     \left( 0 \!\oplus\! e^{-i\theta/2\,\sigma_j} \right) 
	                                                   X 
	                                                  \left( 0 \!\oplus\! e^{i\theta/2\,\sigma_j} \right)
                    	                           \right),  \\
	                           &\pushright{j\in\{1,2,3\}}, \\
	\widetilde\E^{(t)}(X)   &= (1-\gamma) \sum_{i=0}^{1} X_{ii} \otimes \ketbra{1-i}{1-i}.
\end{align*}
One can readily check that the quotient realization of this is the same as for $\E^{(\uu)}$.
\end{example}

\section{Semidefinite representable cones and quotient operator systems}
\label{sec:SDR+quotients}
The CPRP is essentially the problem of inverting the quotient construction presented in the previous section. This corresponds to providing a completely positive lifting of a regular realization $\mathcal R=(\V, \pi,D^{(u)},\tau)$. As is the case in the classical context, a necessary and sufficient condition
will turn out to be the existence of a stable cone of a certain kind. 
We focus on finite-dimensional liftings from an $r$-dimensional regular realization $\mathcal R$ acting on $\V\cong\mathbb{R}^r$ to a completely positive realization acting on $\BH^\sa$ where $\H$ is a finite-dimensional Hilbert space, $\H=\mathbb{C}^n$. We use $\PSD$ to denote the positive semidefinite cone in $\BH$, which
is generating for $\BH^\sa$ (all cones we deal with are convex). 
A cone $\C$ is \emph{pointed} iff $x\in \C$ and $-x\in \C$ implies $x=0$, 
it is \emph{generating} if $\vecspan\,\C=\V$ and is \emph{proper} if it is 
both pointed and generating. We will use calligraphic letters for subspaces of $\BH^\sa$, 
and for any given subspace $\W$, $\W^+$ will denote its intersection with $\PSD$, 
$\W^+=\W\cap\PSD$.

\begin{definition} 
Let $\V$ be a finite dimensional real vector space. 
A \emph{\textbf{semidefinite representable cone}} (SDR) is a set $\C\subset\V$ such that  
\begin{equation}\label{states}
	\C=L(\W^+)
\end{equation}
where $\W\subseteq \B(\mathbb{C}^d)^\sa$ is a subspace, for some $d$, 
and $L:\W\rightarrow \V$ is a linear map. 
\end{definition}

It is easy to see that proper SDR cones can always be described by subspaces $\W$ such that $\W=\vecspan(\W^+)$ and $L$ is a quotient map from $\W$ to $\W/\K\cong \V$, with $\K\cap \PSD=\{0\}$. SDR cones are homogeneous instances of semidefinite representable sets, 
the feasibility regions of semidefinite programs~\cite{blekherman_semidefinite_2012}. 

As we will be dealing with various quotients and the quotient construction used to represent a given SDR will play a relevant role, we find it convenient to introduce the following notation. Given a finite-dimensional Hilbert space $\H=\mathbb{C}^n$, and a subspace $\W\subseteq \BH^\sa$, the cone $L(\W^+)$ will also be denoted by $\C=L(\W^+)=\fracc{\W^+}{\ker L}$ whenever $\ker L\subseteq \W$. Thus, we will often encounter expressions of the form $\fracc{A^+}B$ which imply $B\subseteq A$.

\begin{example}
\label{ex:SDR3}  
We have already encountered a trivial instance of SDR cone in Example~\ref{ex:SDR1}, 
which is isomorphic to the positive semidefinite cone in the subspace of 
self-adjoint operators $\PSD\subset \B(\mathbb{C}^2)^\sa$. We derive the relevant 
cone for Example~\ref{ex:twoqubits}. Consider the spaces $\BH$ and $\BH^*$, 
where $\H=\mathbb{C}^2\otimes \mathbb{C}^2$. Arbitrary elements in $\W$ are 
of the form Eq.~\eqref{eq:arbitraryomega}, which can be rewritten as
\begin{align}
	\omega%=&\,c_0 \mathbb{1}+ \vec n_A \cdot \vec \sigma_A+\vec n_B\cdot \vec \sigma_B\\
	=&\,(c_0-|\vec n_A|-|\vec n_B|)\mathbb{1}\otimes\mathbb{1}+2|\vec n_A|\ket{\hat n_A}\bra{\hat n_A}\otimes\mathbb{1}+2|\vec n_B|\mathbb{1}\otimes\ket{\hat n_B}\bra{\hat n_B},
\end{align}
where $\ket{\hat n}\bra{\hat n}=(\mathbb{1}+\hat n\cdot\vec \sigma)/2$. The condition $\omega\geq0$ then reads
\begin{align}\label{ex:positivity}
	c_0\geq|\vec n_A|+|\vec n_B|.	
\end{align}
%Alternatively, we can write
%\begin{align}
%	\omega=c_0\mathbb{1} +(\vec n_A-\vec n_B)\cdot \frac{\vec \sigma_A-\vec \sigma_B}{2}+(\vec n_A+\vec n_B)\cdot \frac{\vec \sigma_A+\vec \sigma_B}{2}.
%\end{align}
Under the quotient map $L:\W\rightarrow\W/\K$ defined in Example~\ref{ex:twoqubits}, an arbitrary element $x\in\W/\K$ of the form $x=\sum_\mu x_\mu\e_\mu$ corresponds to elements $\omega=x_0\mathbb{1}+\tfrac12\vec x \cdot (\vec \sigma_A-\vec \sigma_B)+k$, where $k\in \K$. 
Thus we have $x\in L(\W^+)$ if and only if there is $k=\tfrac12\vec s \cdot (\vec \sigma_A+\vec \sigma_B)\in \K$ such that 
\begin{align}
	x_0\mathbb{1}+\vec x \cdot\frac{\vec \sigma_A-\vec \sigma_B}{2}+\vec s \cdot\frac{\vec \sigma_A+\vec \sigma_B}{2}\geq0,
\end{align}
Hence the cone $\C=L(\W^+)$ can be cast as an SDR cone in the form
\begin{align}
	L(\W^+)=\left\{x\in\V:\exists s_j \text{ s.t.~}x_0\mathbb{1}+\sum_i x_i\tfrac{1}{2}(\sigma_A-\sigma_B)_i+\sum_j s_j \tfrac{1}{2}(\sigma_A+\sigma_B)_j\geq0\right\}.
\end{align}

However, in this case, the expression can be further simplified. Using condition~\eqref{ex:positivity}, we have
\begin{align}
	x_0\geq\frac{1}{2}( |\vec s+\vec x|+|\vec s-\vec x|),
\end{align}
which can be fulfilled if and only if $x_0\geq|\vec x|$. Thus we can write
\begin{align}
	\C=L(\W^+)=\{x\in\V:x_0\geq|\vec x|\},
\end{align}
which coincides exactly with the cone in Example \ref{ex:SDR1}.
\end{example}

This example reveals two relevant aspects of SDR cones and their representations. 
First, a given SDR cone has an infinite number of equivalent representations. 
We have shown two for the positive semidefinite cone $\PSD\subset\B(\mathbb{C}^2)^\sa$, 
but it is easy to come up with several others. On the other hand, it is not the 
case that every SDR cone can be understood as a positive semidefinite cone, 
or a \emph{slice}, thereof (\emph{i.e.}~a $\W^+$).

\begin{lemma}
\label{LEMMA:DUAL}
Let $\id\in \W\subseteq\BH^\sa$ and $\widetilde\W \subseteq (\BH^\sa)^*$.  % Let $L$ be the quotient map $L:\W\rightarrow \W/\K$. 
The spaces $\big(\W/(\W\cap\widetilde \W^\perp)\big)^*$ and $\big(\widetilde \W+\W^\perp\big)/\W^\perp$ are naturally isomorphic.
If $\W=\vecspan(\W^+)$ and $\big(\W\cap\widetilde \W^\perp\big)^+=\{0\}$ then
\begin{enumerate}
	\item The cone
	  $\displaystyle{\C=\fracc{\W^+}{\W\cap\widetilde \W^\perp}}$
	  is a proper SDR cone.
	\item The dual cone of $\C$ is given by
      \begin{equation}
        \label{forms}
	    \C^*=\fracc{\big(\widetilde\W+\W^\perp\big)^+}{\W^\perp}.
      \end{equation}
\end{enumerate}
%where $\widetilde L$ is the canonical projection $\widetilde L:\widetilde \W+\W^\perp\rightarrow (\widetilde \W+\W^\perp)/\W^\perp\cong \V^*$.
\end{lemma}
The proof is given in Appendix~\ref{app:dual}.

It is convenient to denote the dual projection as $\widetilde L:\widetilde\W+\W^\perp\rightarrow \fracc{(\widetilde\W+\W^\perp)}{\W^\perp}$, under which $\C^*$ reads
\begin{align}
	\C^*=\widetilde L\big((\widetilde\W+\W^\perp)^+\big).
\end{align}
This map will be useful in the following section.

\medskip
In the definition of $\C$ (Eq.~\eqref{states}), the subspace $\W$ and the map $L$ 
appear as mere artifacts to actually describe the set $\C$. As is shown in the 
conditions of Lemma~\ref{LEMMA:DUAL}, not all representations of $\C$ are 
convenient to easily describe the dual $\C^*$. However, at this level this 
appears as a mere choice of convenience (chose $L$ such that its kernel has 
no positive semidefinite elements).
%, so that essentially strong duality holds
However, as we now discuss, the SDR representation does contain much more 
information than the actual SDR set. This information can be unpacked by 
considering the operator system it generates. We refer the reader 
to \cite{paulsen_completely_2003} for a clear exposition of operator systems theory.

Since  $\id\in \W\subseteq\BH^\sa$, $\W$ can be regarded as an operator 
system~\cite{paulsen_completely_2003}. 
Let $\W_n=\W\otimes\B(\mathbb{C}^n)^\sa$ and $\W_n^+$ be its positive cone. 
Likewise, given a linear map $L:\W\rightarrow\V$, let 
$L_n\equiv L\otimes\I_n:\W_n\rightarrow \V_n$. We define cones
\begin{equation}
  \C_n=L_n\big(\W_n^+\big)\subset\V_n.
\end{equation}
When $\K\cap\PSD=\{0\}$, $\K$ is called a \emph{completely order proximinal} 
kernel. This guarantees that $(\V, \C_n,L(\id))$ is a quotient operator 
system~\cite{kavruk_quotients_2013}. Two operator systems $(\V,\C_n,L(\id))$ 
may be different despite the fact that their first level of the hierarchy 
$\C_1\equiv\C$ may be the same. This is illustrated in the following
Example~\ref{ex:differentOS}. 
Let us first consider the duals of $\C_n$. Since $\ker L=\K$, we then have 
$\ker L_n=\K_n$. This corresponds to the quotient $\fracc{\W_n}{\K_n}$. 
However, 
\[
  \K_n = (\widetilde\W^\perp\cap\W)\otimes\B(\mathbb{C}^n)^\sa
       = \widetilde\W^\perp_n\cap\W_n.
\]
Thus we have
\begin{align}
	\C_n   &=\fracc{\W_n^+}{(\widetilde\W_n\cap \W_n)},               \\
	\C_n^* &=\fracc{(\widetilde \W_n+\W_n^\perp\big)^+}{\W_n^\perp}.
\end{align}

%\begin{lemma}[\cite{kavruk_quotients_2013}] 
%\end{lemma}

\begin{example}
\label{ex:differentOS} 
We now consider the two different representations for the positive semidefinite cone 
$\PSD$ as a subset of $\V = \B(\mathbb{C}^2)^\sa \cong \mathbb{R}^4$. 
Let us denote them as $\C$ and $\C'$,
\begin{align}
	\C  &= \B(\mathbb{C}^2)^+,\\
	\C' &= L'(\W^+),
\end{align}
where $\B(\mathbb{C}^2)^+=\PSD$ and $\W^+$ are the positive elements in 
\[
  \W =      \{ Y \oplus Z\ :\ Z=\Phi(Y),\ Y \in \B(\mathbb{C}^2)^\sa \} 
    \subset \B(\mathbb{C}^2\oplus\mathbb{C}^2)^\sa,
\]
as defined in Eq.~\eqref{eq:W-space}, Example~\ref{ex:qubit+qubit}.
(Recall that $\Phi$, up to a unitary conjugation, is the transpose map). 
Furthermore, $L':\W \rightarrow \V$ is defined as in that example,
\[
  L'(Y\oplus \Phi(Y)) = Y,
\]
which is clearly an isomorphism, mapping $\W^+$ to $\C'=\C$.
\emph{I.e.}, these two representations yield exactly the same cone.
Their extensions, however, are different:
\begin{align}
	\C_n  &= \Big(\B(\mathbb{C}^2\otimes\mathbb{C}^n)^+\Big),\\
	\C'_n &= (L'\otimes\id_n)\Big(\big(\W\otimes\B(\mathbb{C}^n)\big)^+\Big) \\
	      &= \{ Y \in \B(\mathbb{C}^2\otimes\mathbb{C}^n) : Y \geq 0 \text{ and } Y^\Gamma \geq 0 \},
\end{align}
where $\Gamma$ denotes the partial transpose (transpose on the first system).
In fact, $\C_n \subsetneq \C_n'$ for all $n > 1$, the inclusion being
clear, and the strictness enough to observe for $n=2$. Indeed, $\C_n$
contains entangled rank-one projections, which are all excluded in $\C_n'$.
It is well-known that $\C_2'$ consists of exactly the separable semidefinite
operators~\cite{peres96, Horodecki1996}.
\end{example}

\section{Regular quasi-realizations as quotient realizations}
\label{sec:regular-is-quotient}
From Theorem \ref{THM:ISOMORPHISM} it follows that if a regular quasi-realization 
$\mathcal R=(\V, \pi,\D^{(u)},\tau)$ has an equivalent completely positive 
realization $\mathcal Q=(\BH^\sa,\rho,\E^{(u)},\id)$, the former ($\mathcal R$) 
must be a quotient realization of the latter ($\mathcal Q$). This implies 
several constraints on the structure of the stable subspaces of $\mathcal Q$, 
and provides necessary conditions for the feasibility of the CPRP. 
The quotient construction was illustrated by Examples~\eqref{ex:twoqubits} and 
\eqref{ex:qubit+qubit}, 
where it is shown how the regular quasi-realization arises naturally from the 
quotient map. In this section we develop the opposite construction, namely, 
deriving a completely positive realization from the quasi-realization. 
To do so, we first illustrate the requirements in a particular case, where 
Arveson's extension theorem plays a central role.

\begin{lemma}\label{lemma:injective}
Let $\mathcal R$ be a quasi-realization $(\W,\pi,D,\tau)$, where $\W$ is a 
selfadjoint subspace $\W\subseteq\BH^\sa$. If the linear maps $D^{(u)}$ satisfy
\begin{enumerate}
	\item $D^{(\uu)}_n\W_n^+\subseteq\W_n^+$ $\forall n\in\mathbb{N}$ ($D^{(\uu)}$ are 
completely positive in the operator system $\W$)
	\item $\pi\in(\W^+)^*$ and
	\item $\tau\in\W^+$ is strictly positive ($\tau>0$),
\end{enumerate}
then $\mathcal R$ can be extended to an 
equivalent, completely positive realization $(\BH^\sa$, $\rho,\F,\id)$.
\end{lemma}

\begin{proof} 
%Let us first show that $\tau\in\W^+$ is strictly positive. Suppose 
%$\mathcal Q=\ker \tau \subset \H$ is nontrivial and let $\mathcal{S}=\mathcal Q^\perp=\mathrm{range}(\tau)$ be its orthogonal complement, so that $\tau|_\S\in\B(\mathcal{S})^+$ is strictly positive. Since $\tau\geq0$ and $\tau=\sum_{\uu\in\M^\ell}D^{(\uu)}(\tau)$, we have that $D^{(\uu)}(\tau)\in\B(\mathcal{S})^+$ for all $\uu\in\M^*$.

%{\color{red} This implies that $D^{(\uu)}\B(\mathcal{S})^+\subseteq \B(\mathcal{S})^+$\\
%\em Proof:
%\begin{enumerate}
%	\item There is an $\epsilon$-ball around $\tau$ that is contained in the conic hull of $\{D^{(\uu)}(\tau),\uu\in\M^*\}$ (follows from the quasi-realization being regular) 
%	\item 
%\end{enumerate}
%}
%, hence the quasi-realization 
%$\mathcal{R}' = (\B(\mathcal{P})^\sa, \pi|_{\B(\mathcal{S})}, D|_{\B(\mathcal{S})}, \tau|_{\B(\mathcal{P})})$
%is a an equivalent to $\mathcal{R}$ but of smaller dimension. 
%Since $\mathcal R$ is regular, this cannot be, so we conclude $\ker\tau=\{0\}$ and thus $\tau>0$.

Since $\tau\in\W$ is strictly positive, then $\W\subseteq\BH$ is an operator system. Then every map $D^{(\uu)}:\W\rightarrow\BH$ is completely positive and by virtue of Arveson's extension theorem, there are completely positive maps $\E^{(\uu)}:\BH\rightarrow\BH$ such that $D=\E|_\W$, hence $D^{(\uu)}(\tau)=\E^{(\uu)}(\tau)$. Since the range of $D$ is $\W$, then $\E^{(\uu)}\W\subseteq\W$. In addition, since $\pi\in(\W_+)^*$ we have, by virtue of Lemma~\ref{LEMMA:DUAL} that $\W^+\equiv\fracc{(\widetilde\W+\W^\perp)^+}{\W^\perp}$. Therefore, under the canonical isomorphism, $\pi=\chi+\W^\perp$ for some positive semidefinite $\chi\in(\widetilde\W+\W^\perp)^+\subset \BH^*$. Thus
\begin{align}
	p(\uu)&=\pi\circ D^{(\uu)}(\tau)\\
		&=(\chi+\W^\perp)\circ\E^{(\uu)}(\tau)\\
		&=\chi\circ\E^{(\uu)}(\tau),
\end{align}
where in the second line we have omitted $\W^\perp$ because $\E^{(\uu)}(\tau)\in\W$. We now define the completely positive map $T(X)=\sqrt \tau X \sqrt \tau$, thus having $\tau=T(\id)$, and can define maps $\F^{(\uu)}=T^{-1}\E^{(\uu)}T$ and $\rho=T^*(\chi)\geq0$ to obtain
\begin{align}
	p(\uu)&=\chi\circ T\circ \F^{(\uu)}(\id)\\
		&=\rho\circ \F^{(\uu)}(\id).
\end{align}
In addition, 
\begin{align}
	\sum_{u\in \M}\F^{(u)}(\id)&=\sum_{u\in\M}T^{-1}\E^{(u)}T(\id)\\
	&=\sum_{u\in\M}T^{-1}\E^{(u)}(\tau)\\
	&=\sum_{u\in\M}T^{-1}D^{(u)}(\tau)\\
	&=T^{-1}(\tau)\\	
	&=\id,	
\end{align}
and similarly $\rho\circ\sum_{u\in\M}\F^{(u)}=\rho$.
\end{proof}

We have seen that the quasi-realization in Example~\ref{ex:quasirealization} can be associated to the positive semidefinite cone $\C=\PSD$ in various ways. Example~\ref{ex:invertiblequbit} shows a realization in $\B(\mathbb{C}^2)$. However, such realization is not completely positive as it involves the transposition map. Example \ref{ex:SDR1} illustrates how the SDR cone $\PSD$ arises naturally as a stable cone of realization \ref{ex:invertiblequbit}, as required in Lemma~\ref{nonnegative}.

We have shown a completely positive realization on two qubits in $\B(\mathbb{C}^2\otimes\mathbb{C}^2)$ (as in Example~\ref{ex:twoqubits}), and in Example \ref{ex:SDR3} we derived the stable cone associated to it. On the other hand, Example~\ref{ex:qubit+qubit} shows how the same process can be described in a subspace $\W\subset\B(\mathbb{C}^4)$, with completely positive maps. This, combined with Lemma~\ref{lemma:injective} guarantees that extensions to the entire $\B(\mathbb{C}^4)$ exist. Examples of such extensions are provided in the same example Example~\ref{ex:qubit+qubit}.

Lemma~\ref{lemma:injective} establishes that once a completely positive realization is given in some $\W\subseteq\BH$ (a concrete operator system), then a completely positive realization on the entire $\BH$ follows naturally from Arveson's theorem. This may lead to think that once a quasi-realization is presented, then the only task is to present the right embedding into a suitable subspace of $\BH$. However, it may occur that the vector space $\V$ of the regular quasi-realization is not isomorphic to any operator system $\W$ on which $D^{(\uu)}$ is completely positive. This is the case when, for example, the regular quasi-realization involves a nontrivial quotient. As we have seen, the resulting positive cones in $\V_n$ are not those of a concrete operator system, and Arveson's extension theorem does not apply. The main contribution of this work is to characterize the conditions which replace those in Lemma~\ref{lemma:injective} when complete positivity w.r.t. a concrete operator system cannot be established.

For a hypothetical completely positive realization $(\BH,\rho,\F,\id)$, the accessible subspace 
$\W=\vecspan\{\F^{(u)}(\id)\}$ is an operator system $\W\subseteq\BH$, and complete positivity 
of $\F$ in $\W$ follows from complete positivity in $\BH$,
\begin{align}
	\F_n(\W_n^+)\subseteq\W_n^+.
\end{align}

The null space $\widetilde \W^\perp $ of $\mathcal Q$, and its restriction to $\W$, $\K=\W\cap \widetilde \W^\perp $ must also be stable under the action of $\F^{(\uu)}$. The quotient space is $\V=\W/\K$ and the canonical projection $L:\W\rightarrow \V$ brings $\mathcal Q$ to $\mathcal R$. Under the quotient construction, the induced maps satisfy the relation
\begin{equation}\label{eq:commutative}
	D\circ L = L\circ \F|_\W.
\end{equation}
In addition, we have the following relations
\begin{subequations}\label{taupi}
\begin{align}
	\tau       &= L(\id),\\
	\pi &= \tilde L(\rho),
\end{align}
\end{subequations}
which relate $\mathcal R$ to $\mathcal Q$. With this, we have $\rho|_\W=\pi\circ L$. Using the definitions of the previous section ($\C_n=L_n(\W_n^+)$), we have
\begin{equation}\label{eq:cpabstract}
	\D_n(\C_n)\subseteq \C_n,\quad\forall n\geq1.
\end{equation}
This is precisely the condition of complete positivity in the quotient 
operator system $(\V, \C_n,L(\id))$. Hence a necessary condition for the 
existence of a CP realization in $\BH$ is that the regular realization 
is completely positive with respect to a quotient operator system, 
together with the relations 
\begin{align}\label{eq:ccstar}
	\tau&\in\C,&\pi&\in\C^*,
\end{align}
which follow from \eqref{taupi}. It can be noticed that condition \eqref{eq:cpabstract} 
together with conditions \eqref{eq:ccstar} are the straightforward generalization 
of the conditions in Lemma~\ref{lemma:injective}, from the concrete operator system 
to abstract operator systems. However, quotient operator systems $\V\simeq\fracc{\W}{\K}$ 
(with $\W\subset\BH^\sa$) are not injective in $\BH^\sa$, and Arveson's theorem does not 
apply. Therefore, these conditions are necessary but not sufficient. In fact, 
there exist completely positive maps in $\V\simeq\fracc{\W}{\K}$ which are not 
induced maps of completely positive maps in $\W\subset\BH^\sa$. 
Therefore, condition 1 in Lemma~\ref{lemma:injective} is too weak to ensure existence 
of a completely positive lift from $\V$ to $\W\subset\BH^\sa$ (equivalently,
by Arveson's theorem, directly to $\BH^\sa$). The following is a concrete example for this phenomenon, which we learned from Vern Paulsen~\cite{Vern:personalcomm}.

\begin{example}
\label{ex:Verno-grazia}
For $\mathcal{H}=\mathbb{C}^{m}$, consider the subspace $\W \subset \BH^\sa$ 
of tridiagonal $m\times m$-matrices, i.e.
\[
  \W = \{ A : A_{jk} = 0 \text{ for } |j-k|>1 \},
\]
and in it the subspace
\[
  \K = \{ A\in\W : A = \operatorname{diag}(A_{jj}:j=1,\ldots,m) \text{ and } \tr A = 0 \}
\]
of traceless diagonal matrices. Clearly, $\K$ is the kernel of a 
completely positive (unital and trace preserving) map, so the quotient
$\fracc{\W}{\K}$ is a bona fide operator system. In fact,
it is a very interesting one, since it was shown in
\cite[Thm.~4.2]{farenick_operator_2012} that $\fracc{\W}{\K}$ is completely 
order isomorphic to
\[
  \S_{m-1} := \operatorname{span}\{I,u_1,\ldots,u_{m-1},u_1^*,\ldots,u_{m-1}^*\} \subset C^*(F_{m-1}),
\]
where $C^*(F_{m-1})$ is the $C^*$-algebra generated by the free universal
unitaries $u_1,\ldots,u_{m-1}$. The latter naturally presupposes a
representation on an infinite dimensional Hilbert space.
Furthermore, the map $L:\W \longrightarrow \S_{m-1}$ given by
\begin{align*}
  L(\ketbra{j}{j+1}) &= \frac1m g_j, \\
  L(\ketbra{j+1}{j}) &= \frac1m g_j^*, \\
  L(\ketbra{j}{J})   &= \frac1m I,
\end{align*}
is in fact completely positive, has kernel $\K$, and the
induced quotient map $\tilde{L}:\fracc{\W}{\K} \longrightarrow \S_{m-1}$
is a complete order isomorphism.

Now, for a permutation $\pi$ of $m-1$ objects, observe that its action
on the generators, $u_j \longmapsto u_{\pi(j)}$, extends to an
automorphism of $C^*(F_{m-1})$, which clearly leaves $\S_{m-1}$
invariant, and hence defines a completely positive and unital map,
in fact an automorphism, $D^{\pi}:\S_{m-1} \longrightarrow \S_{m-1}$,
which we identify with an automorphism of $\fracc{\W}{\K}$.
In particular, $D^{\pi}_n = D^\pi \otimes \id_n$ maps each positive
cone $\left(\fracc{\W_n}{\K_n}\right)^+$ to itself.

For concreteness, let us now look at $m=4$ and the permutation
$\pi=(1)(23)$. We claim that $D^{\pi}$ is not obtained by
taking the quotient of even a positive Hermitian-preserving
map $\F:\W \longrightarrow \W$ that preserves $\K$, 
let alone a completely positive one.
Indeed, such an $\F$ would have to map
\begin{align*}
  \F(\ketbra{1}{2}) &= \ketbra{1}{2} + K_1, \\
  \F(\ketbra{2}{3}) &= \ketbra{3}{4} + K_2, \\
  \F(\ketbra{3}{4}) &= \ketbra{2}{3} + K_3, 
\end{align*}
with traceless diagonal matrices $K_1$, $K_2$, $K_3$. Similarly,
\[
  \F(\ketbra{j}{j}) = R_j,
\]
where $R_j$ are non-negative diagonal matrices with unit trace
($j=1,\ldots,4$). 
Now, let us apply $\F$ to the positive semidefinite matrix
$\ketbra{1}{1}+\ketbra{2}{2}+\omega\ketbra{1}{2}+\omega^*\ketbra{2}{1}$,
with a complex unit $\omega$ (we will only need $\pm 1$ and $\pm i$):
by linearity and positivity,
\[
  R_1 + R_2 + \omega K_1 + \omega^* K_1^* + \omega\ketbra{1}{2}+\omega^*\ketbra{2}{1} \geq 0.
\]
As the trace of the diagonal and necessarily positive semidefinite
$R_1 + R_2 + \omega K_1 + \omega^* K_1^*$ equals $2$, this can only be if
for all $\omega$,
\[
  R_1 + R_2 + \omega K_1 + \omega^* K_1^* = \ketbra{1}{1} + \ketbra{2}{2},
\]
and hence $K_1=0$ and
\[
  R_1 + R_2 = \ketbra{1}{1} + \ketbra{2}{2}.
\]
Analogously, 
working on $\ketbra{2}{2}+\ketbra{3}{3}+\omega\ketbra{2}{3}+\omega^*\ketbra{3}{2}$,
we find $K_2=0$ and
\[
  R_2 + R_3 = \ketbra{3}{3} + \ketbra{4}{4},
\]
but this is a contradiction, as it says that $R_2$ would have to
be supported on $\operatorname{span}\{\ket{1},\ket{2}\}$ and at the same 
time on the orthogonal subspace $\operatorname{span}\{\ket{3},\ket{4}\}$.

More generally, for arbitrary $m$, one can show by a careful 
application of the above reasoning that the only permutations
$\pi$ (viewed as cp maps of $\S_{m-1}$) admitting a lifting
to a positive map $\F:\W \longrightarrow \W$ that preserves $\K$
are those that map neighbours to neighbours, i.e. $|j-k|\leq 1$
implies $|\pi(j)-\pi(k)|\leq 1$; these are precisely the identity 
and the complete reversal $j \leftrightarrow m+1-j$. For the latter,
the lifting $\F$ is unique among the positive maps, however it is 
not completely positive as it is unitarily equivalent to the transpose map.
\end{example}

\medskip
To overcome this problem, we will not impose complete positivity in 
the operator systems sense, but instead a stronger condition that 
guarantees complete positivity in the quotient operator 
system $\V$ as well as in $\W$.
%
%\begin{definition} 
%Let $\mathscr W\subseteq \B(\H\otimes\H)$.
% We say that $\mathscr W$ is triangular if $\chi(\mathscr W)$ is 
% closed under composition, where $\chi:\B(\H\otimes\H)\rightarrow\BH\otimes\BH^*$ 
% is its image under the Choi-Jamio\l{}kowski isomorphism. 
%\end{definition}
%
%Triangular subspaces correspond to Choi matrices of triangular maps. Closedness under composition then implies that for all $A, B\in\mathscr W$ then
%\begin{align}
%	\tr_2[(A\otimes \id)^{\top_2}(\id\otimes B)]\in\mathscr W.
%\end{align}

\begin{lemma} Let $\W\subseteq\BH$ and $\widetilde \W\subseteq\BH^*$. 
%Let $L$ and $\widetilde{L}$ be the quotient maps
%\begin{eqnarray}
%	L:			&\W&\longrightarrow\fracc{\W}{(\W\cap\widetilde\W^\perp)}\equiv\V\\
%	\widetilde{L}:	&\widetilde\W+\W^\perp&\longrightarrow\fracc{(\widetilde\W+\W^\perp)}{\W^\perp}\cong \V^*.
%\end{eqnarray}
Let $\mathscr S^\CP$ be the set of completely positive maps $\F:\BH\rightarrow\BH$ such that
\begin{subequations}\label{eq:stability}
\begin{align}
\label{eq:stabilityW}
	\F(\W)&\subseteq\W\\
\label{eq:stabilitytildeW}
	\F^*(\widetilde \W)&\subseteq \widetilde \W.
\end{align}
\end{subequations}
Then $\mathscr S^\CP$ is a pointed SDR cone in $\BH\otimes\BH^*$.
\end{lemma}
\begin{proof}

Let $\F:\BH\rightarrow\BH$ satisfy Eqs.~\eqref{eq:stability} and regard it as an element in $\BH\otimes \BH^*$. Let $i:\W\hookrightarrow\BH$ be the natural inclusion, and $i^*:\BH^*\rightarrow\W^*$ be its adjoint, \emph{i.e.}, the canonical projection $i^*:\BH^*\rightarrow\fracc{\BH^*}{\W^\perp}\cong\W^*$, with $\ker i^*=\W^\perp$. From Eq.~\eqref{eq:stabilitytildeW} we have that $\F(\widetilde \W^\perp)\subseteq\widetilde \W^\perp$. Hence
\begin{align}\label{eq:stableK}
	\F(\W\cap\widetilde \W^\perp)\subseteq\W\cap\widetilde \W^\perp.
\end{align}

The restriction of $\F$ to $\W$ is given by $\F|_\W=\mathbb{1}\otimes i^*(\F)\in \BH \otimes \W^*$. However, since $\F(\W)\subseteq\W$ we have that in fact 
\begin{align}
	\mathbb{1}\otimes i^*(\F)\in\W\otimes\W^*.
\end{align}
Which implies that $\F\in \W\otimes \BH^*+\BH\otimes \W^\perp$.

Similarly, apply the canonical projection $L\otimes\mathbb{1}$ to obtain $L\otimes i^*(\F)\in \W/\K\otimes \W^*$, with $\K=\W\cap\widetilde \W^\perp$. However, from Eq.~\eqref{eq:stableK} we have that $L\otimes i^*(\F)(\K)=0$, thus $L\otimes i^*(\F)\in \fracc{\W}{\K}\otimes \K^\perp$. This determines that 
\begin{align}
	\mathbb{1}\otimes i^*(\F)\in \W\otimes \K^\perp+\ker L\otimes\W^*.
\end{align} 
However, $i^*{}^{-1}(\K^\perp)=\W^\perp+\widetilde \W+\ker i^*$, while in general, 
\begin{align}
	\mathbb{1}\otimes i^*{}^{-1}( \W\otimes \K^\perp)\in \W\otimes (\W^\perp+\widetilde \W)+\BH\otimes\ker i^*.
\end{align}
Therefore, maps satisfying $\F(\W)\subseteq\W$ and $\F(\widetilde \W)\subseteq \widetilde \W$ are in the subspace $\mathscr S\subseteq \BH\otimes\BH^*$,
\begin{align}
	\mathscr S=\W\otimes (\W^\perp+\widetilde \W)+(\W\cap\widetilde \W^\perp)\otimes \BH^*+\BH\otimes \W^\perp,
\end{align}
where we have used $\ker L=(\W\cap\widetilde \W^\perp)$ and $\W^\perp=\ker i^*$.

To put this expression in perspective, we refer the reader back to Lemma~\ref{LEMMA:DUAL}, where it was pointed out that the spaces $\fracc{\W}{(\W\cap\widetilde \W^\perp)}$ and $\fracc{(\W^\perp+\widetilde\W)}{\W^\perp}$ are natural duals of each other. Then the quotient maps $L:\W\rightarrow\V$ and $\widetilde L:(\W+\widetilde\W^\perp)\rightarrow\V^*$ play a central role in understanding also the maps satisfying relations \eqref{eq:stability}. Indeed, $\mathscr S$ can be written as
\begin{align}
	\mathscr S=\mathrm{dom\,}L\otimes\mathrm{dom\,}{\widetilde{L}}+\ker L\otimes\BH+\BH^*\otimes \ker \widetilde L
\end{align}

Hence, one can effectively understand this subspace as a triangularization condition for the maps $\F$, where the forms in $\mathrm{dom\,}{\widetilde L}$ are coupled to elements in $\mathrm{dom\,}L$, while arbitrary forms are only coupled to $\ker L$ and only $\ker{\widetilde{L}}$ is coupled to arbitrary vectors.

%
%Now, removing redundancies, we obtain
%\begin{align}
%	\F\in\W\otimes \widetilde \W+\K \otimes\BH^*+\BH\otimes \W^\perp.
%\end{align}

Let $\chi$, $\phi=\chi^{-1}$ be the Choi-Jamio\l{}kowski isomorphisms. Under this isomorphism CP maps become positive semidefinite operators in $\B(\H\otimes\H)$, thus $\F$ is CP if and only if $\phi(\F)\geq0$. Hence, if $\F$ is CP we have
\begin{align}
	\F\in\mathscr S^\CP=\mathscr S\cap \chi(\B(\H\otimes\H)^+).
\end{align}
To see that $\mathscr S^\CP$ is SDR, let $\mathscr W=\phi(\mathscr S)\subseteq\B(\H\otimes\H)$. Then,
\begin{align}
	\mathscr S^\CP=\chi(\mathscr W^+).
\end{align}

This shows that the set of CP maps which satisfy the stability relations \eqref{eq:stability} forms an SDR cone. It is easy to see that $\mathscr S^\CP$ is a pointed cone. Since $\chi$ is an isomorphism, it is enough to show that $\mathscr W^+$ is pointed. %Suppose $\pm x\in\mathscr W^+$. Then $\pm \chi(x)$ are both CP, which implies that $\chi(x)=0$. Thus $x=0$. 
In general, for any pointed cone, its restriction to a subspace is also pointed. Hence pointedness of $\B(\H\otimes\H)^+$ implies pointedness of $\mathscr W^+$.
\end{proof}

\begin{lemma}
If $\W$ and $\widetilde \W$ have order units, then $\mathscr W$ has an order unit and hence $\mathscr W^+$ is generating, \emph{i.e.} $\mathscr W^+-(\mathscr W^+)=\mathscr W$.
\end{lemma}
\begin{proof}
Let $\W$ and $\widetilde \W$ have order units $e$ and $\tilde e$, respectively. Then, the map $\E=e\otimes \tilde e$ satisfies Eqs.~\eqref{eq:stability}, and hence is in $\mathscr S$. In addition $\E$ is completely positive, hence $\E\in\mathscr S^\CP$. To show that $\E$ is an order unit we check strict positivity under the Choi-Jamio\l{}kowski isomorphism $\phi(\E)$,
\begin{align}
	\phi(\E)=\sum_{ij}\tilde e(\ket i\bra j)\, \ket i\bra j\otimes e.
\end{align}
To show that this is strictly positive definite, take $\ket v=\sum_{kl}c_{kl} \ket{kl}$ and compute $\bra v\phi(\E)\ket v$, which reduces to
\begin{align}\label{eq:sum}
	\bra v\phi(\E)\ket v=\sum_k \lambda_k\, \tilde e(\ket{u_k}\bra{u_k}),
\end{align}
where $\ket{u_k}=\sum_i u_{ik}\ket i$ and $u_{ik}, \lambda_k$ are the unitaries and eigenvalues of the positive semidefinite matrix
\begin{align}
	M_{ij}=\left(\sum_k c_{ik}^* \bra{k}\right) e\left(\sum_l c_{jl}\ket{l}\right), 
\end{align}
such that $M_{ij}=\sum_k u_{ik}\lambda_k u_{jk}^*$. Since $\tilde e$ is positive definite then $\tilde e(\ket{u_k}\bra{u_k})>0$ for all $\ket{u_k}$. Since $e$ is strictly positive, then $M$ has some nonzero eigenvalue. Hence the sum in Eq.~\eqref{eq:sum} has some strictly positive summand. Thus
\begin{align}
	\phi(\E)>0.
\end{align}
This implies that $\phi(\E)\in\mathscr W$ is an order unit, hence, for any $W\in\mathscr W$ we have $\lambda>0$ such that $W+\lambda\phi(\E)>0$, and $W=(W+\lambda(\E))-\lambda\phi(\E)$, 
thus $\mathscr W=\mathscr W^+-(\mathscr W^+)$.
\end{proof}

In addition, we also point out that $\mathscr S^\CP$ is a semigroup, \emph{e.g.} for any $\E, \F\in\mathscr S^\CP$, $\E\circ \F\in \mathscr S^\CP$. We are now in position to obtain the structure of the set of maps $D:\V\rightarrow\V$ which are induced by completely positive maps in $\BH$. 

%We thus have that $D=L\otimes\widetilde L(\E)\in \V\otimes \V^*$ is the induced quotient map. Also, completely positive maps with these stable subspaces form a cone $\mathscr S^\CP$, where $\CP$ denotes intersection with the completely positive cone. Finally, we conclude that 
%\begin{align}
%  \label{scarrafo}
%  \D\in \scarrafone=L\otimes \widetilde L  \left(\mathscr S^\CP\right).
%\end{align}
%

\begin{theorem}\label{thm:scarrafone} Let $\W\subseteq\BH$, $\widetilde\W\subseteq\BH^*$ be concrete operator systems, such that $\K=\W\cap\widetilde\W^\perp$ has no positive semidefinite elements $\K\cap\PSD=\{0\}$. Then $\V\simeq\fracc{\W}{\K}$ is a quotient operator system. Let $L:\W\rightarrow\V$ be the quotient map.

Let $D:\V\rightarrow\V$. Then, a completely positive map $\F:\BH\rightarrow\BH$ exists such that
\begin{enumerate}
	\item $\F_n(\W_n)\subseteq\W_n$,
	\item $\F_n^*(\widetilde \W_n)\subseteq\widetilde \W_n$, and 
	\item $L\circ \F|_\W=D\circ L$,
\end{enumerate}
if and only if
\begin{align}\label{scarrafo}
	D\in L\otimes\widetilde L(\mathscr S^\CP)
\end{align}
\end{theorem}
\begin{proof}
	The ``if'' part is trivial. To prove the ``only if'' part, recall that $\widetilde L:\W^\perp+\widetilde\W\rightarrow\V^*$, thus the adjoint map is the natural inclusion $\widetilde L^*:\V\hookrightarrow(\W^\perp+\widetilde\W)^*$. This serves the purpose of lifting $\V$ to $\BH$, and thus
\begin{align}\label{eq:projectedquotient}
	D=L\circ \F\circ \widetilde L^*.
\end{align}
One can check that with this, maps in $\mathscr S$ satisfy the commutative relation Eq.~\eqref{eq:commutative},
\begin{align}
	D\circ L=L\circ \F\circ\widetilde L^*\circ L=L\circ \F|_\W ,
\end{align}
as expected for quotient maps. Furthermore, in the $\BH\otimes\BH^*$ notation, we have
\begin{align}
	D=L\otimes\widetilde L(\F),
\end{align}
Hence all quotient maps in $\V$ corresponding to completely positive maps in $\mathscr S$ are contained in the cone
\begin{align}\label{eq:scarrafone}
	\scarrafone =L\otimes \widetilde L(\mathscr S^\CP).
\end{align}

Notice that in principle, the range of $L\otimes \widetilde L$ is not well defined in the entire $\BH\otimes \BH^*$, and arbitrary extensions would be required. However, for each of the subspaces making up $\mathscr S$ it is well-defined,
\begin{align}
	L\otimes \widetilde L:\left\{
		\begin{array}{ll}
			\mathrm{dom\,}L\otimes\mathrm{dom\,}{\widetilde{L}}&\longrightarrow~ \V\otimes\V^*,\\
			\ker L\otimes\BH&\longrightarrow~ 0,\\
			\BH^*\otimes \ker \widetilde L&\longrightarrow~ 0.
		\end{array}\right.
\end{align}
This clarifies that Eq.~\eqref{eq:projectedquotient} is well defined. 
\end{proof}

Another relevant question is to elucidate how the structure of $\scarrafone$ determines the hierarchy $\{\C_n\}$. As we have seen, if the spaces $\W$ and $\widetilde \W$ have order units, then the resulting $\mathscr W^+$ is generating, with order unit $\E$. Then $E=L\otimes\widetilde L(\E)$ is an order unit in  $\V\otimes \V^*$, \emph{i.e.}, for any $D\in\V\otimes\V^*$ we have some $\F\in\mathscr S$ such that $D=L\otimes\widetilde L(\F)$. Then, there is always some $\lambda>0$ such that
\begin{align}
	D+\lambda E=L\otimes\widetilde L(\F+\lambda \E)\in\scarrafone.
\end{align}
Thus, $\V\otimes\V^*=\scarrafone-\scarrafone$. 

Based on Theorem \ref{thm:scarrafone}, one can readily verify that $D\in\scarrafone$ implies $D$ and its adjoint $D^*$ are completely positive in their respective operator systems.

\begin{lemma} Let $D\in\scarrafone$. Then $D_n(\C_n)\subseteq\C_n$ and $D_n^*(\C_n^*)\subseteq\C_n^*$.
\end{lemma}

\begin{proof}Let $D\in\scarrafone$. Then, there is a completely positive map $\F\in\mathscr W^+$ such that $D\circ L=L\circ \F$ (and hence $D_n\circ L_n=L_n\circ \F_n$). Also, for any $x\in\C_n$ there is $\omega\in\W_n^+$ such that $x=L_n(\omega)$. Thus,
\begin{align}
	D_n(x)=D_n\circ L_n(\omega)=L_n\circ \F_n(\omega)\in L_n(\W_n^+).
\end{align}
This also implies that for any $y\in\C_n^*$, and for any $x\in\C_n$, we have
\begin{align}
	D_n^*(y)(x)=y\circ D_n(x)\geq0,
\end{align}
hence $D_n^*(y)\in\C_n^*$.
\end{proof}

With this, we now consider the rank-1 elements in $\scarrafone$. These elements can be regarded as probabilistic preparations, as they produce effects (states) independent of the effect (state) they (their adjoints) act upon. Let the set of rank-1 elements in $\scarrafone$ be denoted as $\mathscr X$. Also, let $\mathcal A, \mathcal B$ be vector spaces and let $A,B$ be respective subsets. Then, we denote by 
\begin{align}
	A\odot B=\{x\in\mathcal A\otimes\mathcal B:\exists a\in A,\exists b\in  B\mathrm{~such~that~} x=a\otimes b\}
\end{align}
the set of tensor products from $A$ and $B$.

\begin{lemma} $\mathscr X=\C^*\odot \C$.
\end{lemma}
\begin{proof} We first show $\C\odot\C^*\subseteq\mathscr X$. Since all elements in $\C\odot\C^*$ are of rank 1, it is enough to show $\C\odot\C^*\subseteq\scarrafone$.
Let $D\in\C\odot\C^*$ be of the form $D=L(x)\otimes \tilde L(y)$ for some $x\in \W^+$ and $y\in(\W^\perp+\widetilde\W)^+$. Then $x\otimes y\in\W^+\otimes(\W^\perp+\widetilde\W)^+$. Then clearly $x\otimes y\in\mathscr S^\CP$, and thus $D=L\otimes\tilde L(x\otimes y)\in\scarrafone$.

We now show $\mathscr X\subseteq\C\odot\C^*$. % Clearly, $\mathscr X$ is closed under composition, \emph{i.e.}, for any $X_1$, $X_2\in\mathscr X$, $X_1X_2\in\mathscr X$. 
Clearly, $X(\C)\subseteq\C$ and $X^*(\C^*)\subseteq\C^*$ for all $X=\mu\otimes\nu\in\mathscr X$. Since $\C\odot\C^*\subseteq\mathscr X$, and $\C$, $\C^*$ are generating, it is possible to chose $X_0=a\otimes b$ with $a\in\C$, $b\in\C^*$ such that $\nu(a)>0$, $b(\mu)>0$. Then, for all $X=\mu\otimes\nu\in\mathscr X$ we have
\begin{align}
	X X_0X=X(a)\otimes X^*(b)\in\C\odot \C^*.
\end{align}
But also $XX_0X=\nu(a)b(\mu) X$. Hence $\nu(a)b(\mu) X\in\C\odot\C^*$, thus $X\in\C\odot\C^*$.
\end{proof}

This result suggests that once $\scarrafone$ is given, the cones $\C$ and $\C^*$ are naturally embedded into it. This is true, but it is not the end of the story. Indeed, the cones $\C$ and $\C^*$ are obtained in two different ways. 
\begin{lemma} Let $\scarrafone=L\otimes \tilde L(\mathscr W^+)$ be an SDR mapping cone of the type Eq.~\eqref{eq:scarrafone}. Then the cones $\C$ and $\C^*$ are projections of $\scarrafone$.
\end{lemma}
\begin{proof} Since $\scarrafone$ is pointed, it has a supporting hyperplane $F:\V\otimes \V^*\rightarrow \mathbb{R}$ such that $F(\scarrafone)>0$. In particular, one of such supporting hyperplanes is given by $F=a\otimes b$ where $a\in\C^*$ (resp. $b\in\C$) is a supporting hyperplane of $\C$ (resp. $\C^*$). Then, the maps $\hat a:\V\otimes\V^*\rightarrow\V^*$ defined as $\hat a(D)=D^*(a)$ and $\hat b:\V\otimes\V^*\rightarrow\V$ as $\hat b(D)=Db$ we have
\begin{align}
	\hat a(\scarrafone)=\C^*\qquad\mathrm{and}\qquad \hat b(\scarrafone)=\C.
\end{align}
We will prove only the first identity. One inclusion is clear,
\begin{align}
	\C^*= a(\C)\C^*=\hat a(\C\odot\C^*)=\hat a(\mathscr X)\subseteq \hat a(\scarrafone).
\end{align}
The other inclusion is shown as follows. Take $D\in\scarrafone$ and $w\in\C$. Then $D(\omega)\in\C$, thus $\hat a(D)(\omega)=D^*(a)\omega=a\circ D(\omega)\geq0$, so $\hat a(D)\in\C^*$. Thus $\hat a(\scarrafone)\subseteq\C^*$. The other relation is proven analogously.
\end{proof}

Finally, the set $\C$, $\C^*$ are also restrictions of $\scarrafone$.
\begin{lemma}  Let $\scarrafone=L\otimes \tilde L(\mathscr W^+)$ be an SDR mapping cone of the type Eq.~\eqref{eq:scarrafone}. Then the cones $\C$ and $\C^*$ are restrictions of $\scarrafone$ to suitable subspaces.
\end{lemma}
\begin{proof} Let $a$ be an order unit of $\C^*$. Then $\scarrafone|_{a\otimes \V^*}=\{D\in\scarrafone|\exists b\in\V^*\mathrm{~s. t.~} D=a\otimes b\}$ is a subset of $\mathscr X$, with elements of the form $a\odot \C^*$. The natural isomorphism $a\otimes\V^*\rightarrow\V^*$ takes this onto $\C^*$. A similar argument shows how $\scarrafone|_{\V\otimes b}=\C\odot b$, with $b$ an order unit of $\C$.
\end{proof}

\section{Characterization of completely positive realizations by SDR mapping cones}
\label{sec:ogni-scarraffon-e-bbello-a-mmamma-suia}
So far we have derived a set of necessary conditions which follow 
from the hypothesis that an underlying completely positive realization 
exists. In this section we show that these are also sufficient.

\begin{proposition}[Removing spurious eigenvectors]
\label{sieve} 
Let $\{\E^{(u)}\}$ be a set of completely positive maps on $\BH$ with 
$\E=\sum_u \E^{(u)}$, and let $\rho$, $\id$ be positive semidefinite operators 
in $\BH$ such that $\tr \rho\id=1$. 
If $\omega$ is a positive semidefinite eigenvector of $\E$ such that 
$\tr \rho\omega =0$, then there is always another set of completely positive maps 
$\{\hat\E^{(u)}\}$ on $\B(\ker(\omega))$ and positive semidefinite operators 
$\hat\rho$, $\hat\id\in\B(\ker(\omega))$ such that
\begin{equation}
	\tr[\rho\, \E^{(\uu)}(\id)]=\tr[\hat \rho\, \hat \E^{(\uu)}(\hat \id)].
\end{equation}
for all $\uu\in\M^*$.
\end{proposition}

\begin{proof}
Let $\mathcal{S}=\ker(\omega)$ and $\mathcal Q=\mathrm{range}(\omega)=\mathcal{S}^\perp$ its orthogonal complement. Let $P$  (resp. $Q$) be the corresponding orthogonal projection in $\H$, and $\P=P\,\cdot\,P$, (resp. $\Q$) the hereditary projection on $\BH$.
Since $\omega$ is a positive semidefinite eigenvector, we have 
that $\E\circ\Q=\Q\circ\E\circ\Q$. From positivity, this extends to all $\E^{(\uu)}$ and thus
\begin{equation}
	\P\circ \E^{(\uu)}=\P\circ \E^{(\uu)}\circ\P,\quad\forall \uu\in\M^*.
\end{equation}
From orthogonality of $\rho\geq0$ and $\omega\geq0$ it follows that $\rho=\P(\rho)$ and we can write
\begin{eqnarray}\nonumber
	p(\uu)&=&\,\tr[\rho\, \P \E^{(\uu)}(\id)]\\
\label{projectedCPR}
	&=&\,\tr[\rho\, \P\E^{(u_1)}\P\circ\P\E^{(u_2)}\P\cdots\P\E^{(u_\ell)}\P(\id)].
\end{eqnarray}
Replace $\H\leftarrow \mathcal{P}$, $\BH\leftarrow \B(\mathcal{P})$ and 
\begin{subequations}
\begin{eqnarray}
	\E^{(u)}&\leftarrow&\P\E^{(u)}\P\\
	\id&\leftarrow &\P(\id)\\
	\rho&\leftarrow&\P(\rho).
\end{eqnarray}
\end{subequations}
The resulting maps are still completely positive and  $\rho$, $\id$ are positive semidefinite with support in $\B(\mathcal{S})$, thus the new $\id$ has $\tr[\id\omega]=0$. In addition, from Eq.~\eqref{projectedCPR}, they generate the same process. 
\end{proof}

\begin{theorem}[``\,'O scarrafone'']
\label{mainresult}
Given a pseudo-realization $\mathcal{R}=(\V,\pi,\D,\tau)$, an equivalent, 
finite-dimensional, unital, completely positive realization $(\BH^\sa,$ $\rho,\E,\id)$ exists
if and only if there is an SDR cone $\scarrafone \subset \V\otimes \V^*$ such that
\begin{enumerate}
	\item $\D^{(u)}\in\scarrafone$ for all $u\in\M$, 
	\item $\tau\in \C$,
	\item $\pi\in \C^*$.
\end{enumerate}
where $\C$, $\C^*$ and $\scarrafone$ are of type \eqref{states}, \eqref{forms} and \eqref{scarrafo}, respectively.
\end{theorem}

\begin{proof}
That the conditions are necessary was proven in the previous section. 
It follows from condition 1 that CP maps $\E^{(u)}:\BH\rightarrow\BH$ 
can be defined such that $\E^{(u)}(\K)\subseteq \K$ and $\E^{(u)}(\W)\subseteq\W$, and that 
\begin{equation}
	L\circ \E^{(u)}=\D^{(u)}\circ L,\quad\forall u\in\M.
\end{equation}

To lift the vectors $\tau$ and $\pi$, notice that since $\tau\in\C$ and $\pi\in\C^*$, there is $\id\in\W^+$ and $\rho\in(\W^\perp+\widetilde \W)^+$ such that 
\begin{eqnarray}
	\tau&=&L(\id)\\
	\rho&=&\pi\circ L.
\end{eqnarray}
At this point it is easy to check that $\D^{(\uu)}(\tau)=\D^{(\uu)} L (\id)=L\E^{(u)}(\id)$, so that
\begin{equation}\label{pseudoCPrealization}
	\pi\cdot\D^{(\uu)}(\tau)=\rho \circ\E^{(\uu)}(\id),\quad\forall\uu\in\M^*,
\end{equation}

However, the operators $\rho$ and $\id$ are not left- and right-eigenvectors of $\E=\sum_{u\in\M}\E^{(u)}$, so they $(\BH^\sa,\id,\E,\rho)$ does not form a realization. In order to find a proper completely positive realization, we will iteratively replace them by suitable projections by making use of Theorem~\ref{sieve}, until the desired properties are obtained. In the process, we remove all spurious contributions to $\rho$ and $\id$ until only relevant contributions to Eq.~\eqref{pseudoCPrealization} remain.

\begin{itemize}[leftmargin=1.7cm]
\item [{\bf STEP 1:}] Consider the Ces\`aro mean $\omega_n=\frac{1}{n}\sum_{k=1}^n \E^k(\id)$. Clearly, $\omega_n\geq0~\forall n$. 

Define the ratio $\lambda= \lim_{n\rightarrow\infty}\frac{\|\omega_{n+1}\|}{\|\omega_n\|}$ so that the limit is well-defined,
\begin{equation}
	\omega=\lim_{n\rightarrow\infty}\frac{ \omega_n}{\lambda^n}.
\end{equation}
 Clearly, $\omega\geq0$, and 
 \begin{equation}
 	\E(\omega)=\lim_{n\rightarrow\infty}\frac{1}{n\lambda^n}\sum_{k=1}^n \E^{k+1}(\id)=\lambda\omega.
\end{equation}
At this point, two different scenarios may occur. Either $\lambda=1$ or $\lambda>1$. 
Consider first the case when $\lambda>1$. This means that there is a contribution to $\id$ which grows under the action of $\E$, and $\omega$ captures its asymptotic behavior. One can see that 
\begin{equation}
	\tr[\rho\omega] = \lim_{n\rightarrow\infty}\frac{1}{n\lambda^n}\sum_{k=1}^n\tr[\rho \E^{k}(\id)]
                    = \lim_{n\rightarrow\infty}\frac{1}{\lambda^n} = 0.
\end{equation}
Hence, by making use of Theorem~\ref{sieve}, we can obtain a new set of CP maps $\{\E^{(u)}\}$, $\rho$ and $\id$ such that $\tr[\id\omega]=0$. However, $\rho$ and $\id$ are still not eigenvectors. Repeat STEP 1 until $\lambda=1$.

If $\lambda=1$ then $\omega=\lim_{n\rightarrow\infty} \omega_n$ is well defined. Replace $\id\leftarrow \omega$ and proceed to STEP 2. 
\end{itemize}

At each iteration of STEP 1 a new $\omega$ is obtained, orthogonal to all previous ones, and the associated eigenvalue can only be equal or decrease. The aim of this iteration is to capture the eigenspace of $\E$ with the largest eigenvalue and remove it without altering the resulting stochastic process $p(\uu)$.

Because $\E$ has only finitely many eigenvalues, eventually $\lambda$ will equal 1. In that  case, the resulting $\omega$ is strictly positive. Proceed to PART~2.

\begin{itemize}[leftmargin=1.7cm]
\item[{\bf STEP 2:}] At this point $\id$ is an eigenvector but $\rho$ is not.  Return STEP 1 with the dual realization, {\em i.e.}, with $((\BH^\sa)^*,\id,\E^*,\rho)$, interchanging the roles of $\rho$ and $\id$.
\end{itemize}

After STEP 2, $\rho$ is an eigenvelue of $\E$ but $\id$ may not be. A further iteration of steps 1 and 2 will lead to further dimension reductions. Since the dimension is finite, eventually no further truncations will be necessary and both $\id$ and $\rho$ will be proper left- and right- eigenvalues of $\E$.

Once one has iterated through STEPS 1 and 2, one has a completely positive realization $(\rho,\E^{(u)},\id)$ with the required stability properties for $\rho$ and $\id$. It just remains to ensure that $\id>0$. The procedure is very similar to the one just exposed.

\begin{itemize}[leftmargin=1.7cm]
	\item [{\bf STEP 3:}] Let $\mathcal Q=\ker(\id)$ and $\mathcal{S}=\mathcal Q^\perp=\mathrm{range}(\id)$ its orthogonal complement. Since $\id\geq0$ is an eigenvector of $\E$, we have that $\E^{(\uu)}(\id)\in\B(\mathcal{S})$ for all $\uu\in\M^*$. Hence we can make the substitutions $\H\leftarrow \mathcal{S}$, $\BH\leftarrow \B(\mathcal{S})$ and 
\begin{subequations}
\begin{eqnarray}
	\E^{(u)}&\leftarrow&\P\E^{(u)}\P\\
	\id&\leftarrow &\P(\id)\\
	\rho&\leftarrow&\P(\rho).
\end{eqnarray}
\end{subequations}
With this, now $\id>0$. One can define the completely positive map $\mathcal N(x)=\id^{-1/2}x\id^{-1/2}$. Finally, replace
\begin{subequations}
\begin{eqnarray}
	\E^{(u)}&\leftarrow&\mathcal N\E^{(u)}\mathcal N^{-1}\\
	\id&\leftarrow& \mathcal N(\id)=\mathbb{1}\\
	\rho&\leftarrow&\mathcal N^{-1}(\rho).
\end{eqnarray}
\end{subequations}
\end{itemize}
This substitution makes $\sum_{u\in\M}\E^{(u)}(\mathbb{1})=\mathbb{1}$, while preserving complete positivity and the resulting $\rho$ is the stationary state of the system. 
\end{proof}

\medskip
Note that several steps in the reduction algorithm could be avoided by imposing further conditions on the properties of the subspaces defining $\scarrafone$, but to explore these relations is beyond the scope of our present work.

The constructive algorithm in the above proof shows that not only appropriate 
completely positive maps can be obtained from the condition $\D\in\scarrafone$, 
but also that their structure can be cast into the form of a quantum instrument, 
and $\rho$ is a fixed point of $\sum_{u\in\M}\E^{(u)}$. The fact that a dimension 
smaller than that of $\BH$ is capable of reproducing the model described by 
$(\BH^\sa,\rho,\E,\id)$ is ultimately due to the non-primitivity of $\E^*$ and 
the lack of information completeness of the POVM elements  
$M^{(\uu)}=\E^{(\uu)}(\id)$. Theorem~\ref{mainresult} establishes that 
this explanation is the only possible one, revealing the essential traits 
that a quasi-realization should exhibit in order to be equivalent to a 
higher-dimensional quantum model.

\section{Discussion}
Our main result, Theorem~\ref{mainresult}, represents a generalization of 
Dharmadhiraki's polyhedral cone condition~\cite{dharmadhikari_sufficient_1963},
and establishes the type of positivity that needs to be respected at the 
level of a quasi-realization for it be completely positively realizable in 
in a certain way $\BH$. 
This brought to light a central issue that goes unnoticed in the commutative 
case. Unlike in the formulation of Dharmadhiraki's cone condition, the truly 
fundamental object is the set of maps $\scarrafone$, from which the cones 
$\C$ and $\C^*$ can be derived. This shifts the focus from the geometry of 
the cone of states to the cone and at the same time semigroup of transformations 
corresponding to a given process $p$.

%\textcolor{red}{EXPLAIN HOW C AND C* ARE RECOVERED FROM P ONCE SOME CONDITIONS ARE MET.}

This is far from a full solution to the problem. Although 
condition \eqref{scarrafo} can be verified by a semidefinite program, finding 
a suitable cone $\scarrafone$ for a given process is still a formidable 
challenge. Our result highlights significant departures from the PRP, so that 
novel approaches may be possible. In particular, the CPRP turns out to be 
deeply related to lifting properties for quotient operator systems. Aspects 
of this theory are deeply connected with several open questions in operator 
theory~\cite{farenick_operator_2012}, such as Connes Embedding Problem and 
Kirchberg's conjecture. In addition, classical algorithms for learning Hidden 
Markov Models using matrix factorizations~\cite{cybenko_learning_2011} may be 
extended to semidefinite factorizations~\cite{fiorini_linear_2012,gouveia_lifts_2013} 
thus establishing links between the computational complexity of the CPRP 
and that of other relevant problems in Quantum Information science. 
An interesting question, from the operator systems theory point of view, 
is to identify the abstract operator system in $\V$ for which $\scarrafone$ 
is precisely the cone of completely positive maps, and to determine its 
nuclearity properties. 

Conversely, for a given quotient operator system
of a concrete operator system in a finite-dimensional Hilbert space, 
there are in fact always infinitely many $\W \subset \BH^\sa$ and
quotient maps $L:\W\rightarrow \V$ realizing the same quotient. But while 
all of these quotients give rise to the same set of completely positive
maps, the cones $\scarrafone$ depend on the pairs $(\W,L)$, but all
contained in the cone of completely positive maps of $\V$. The question
is whether the cone generated by all these $\scarrafone$ together (or 
at least its closure) exhausts all completely positive maps.

Just as the positive realization problem, the completely positive realization 
problem is highly relevant in systems identification and quantum control. 
It addresses the problem of finding compact models for systems with quantum 
memory and a classical readout interface. In particular, modeling stochastic 
processes which are generated by quantum devices will be the primary 
application of our results. The positive description of a process not 
only provides insight into the physical mechanisms underlying a process, 
but allows to identify \emph{latent} variables, \emph{i.e.}~variables 
that are not directly observed but allow to see order and simplicity in 
otherwise apparently chaotic and highly unpredictable behavior. In this 
sense, accounting for hidden quantum mechanical mechanisms, and more 
importantly, quantum memory to an information source, is potentially the 
difference between obtaining a simple description of a process or a highly 
complex one. 

Finally, we draw the reader's attention to the question of how general 
completely positive realizations are, in relation to general quasi-realizations
(we know that they are strictly more powerful than classical positive
realizations). In particular, in the setting we considered here, of
a finite set of maps on a finite-dimensional vector space, can the cone required 
in Lemma~\ref{nonnegative} always be assumed to be SDR? 
More specifically, consider the smallest cone 
$\C = \mathbb{R}_{\geq 0}\{ \D^{(\uu)}\tau : \uu \in \M^* \}$ generated
by a quasi-realization; is $\C$ an SDR cone, or at least semi-algebraic? 
(Note that SDR implies semi-algebraic, and if the Helton-Nie 
conjecture \cite{HeltonNie-2009} is true they are equivalent.) 
Conversely, if that is not the case, it would mean that there
exists a process with a (finite-dimensional) quasi-realization which
cannot be reproduced by any quantum system of finite dimension.

\section*{Acknowledgments}
We thank Miguel Navascu\'es, Robin Blume-Kohout, Vern Paulsen and Mihai Putinar
for insightful discussions on realizations, SDR cones and complete positivity. 
Especially we would like to thank Vern Paulsen for sharing with us
the insights behind Example~\ref{ex:Verno-grazia}.

This research was partly undertaken while AM was at the Centre for Quantum Technologies (CQT), 
National University of Singapore, and while AW was at the Department of Mathematics, 
University of Bristol, and partly affiliated with CQT.
AM and AW were supported by the ERC (Advanced Grant ``IRQUAT'',
contract no.~ERC-2010-AdG-267386) and the Spanish MINECO 
(project FIS2008-01236) with the support of FEDER funds. 
AW furthermore acknowledges support by the EC (STREP ``RAQUEL'',
contract no. FP7-ICT-2013-C-323970).

\appendix

\section{Proof of Lemma~{\ref{LEM:ORDERUNIT}}}  %%%%% Appendix A %%%%%
\label{app:proof1}

%%%%%%%%%%%%%% Restatement of Lemma \ref{LEM:ORDERUNIT} %%%%%%%%%%%%%%
\noindent
\textbf{Lemma~\ref{LEM:ORDERUNIT}.}
Let $\mathcal R=(\V,\pi,D^{(\uu)},\tau)$ be a regular realization, 
\emph{i.e.}, no equivalent quasi-realization exists of smaller dimension. 
Then, any cone $\C$ satisfying conditions in Lemma~\ref{nonnegative} 
is proper: $\C$ does not contain nor is contained in a proper subspace 
of $\V$, and $\tau$ is an order unit of $\C$.
%%%%%%%%%%%%%%%%%%%%%%%%%%%%%%%%%%%%%%%%%%%%%%%%%%%%%%%%%%%%%%%%%%%%%%

\medskip
\begin{proof} Suppose $\C$ is not proper. This means it is either not pointed or not generating. Consider first the case $\C$ is not generating. Then $\W=\C-\C$, is a proper subspace of $\V$, and from conditions 1 and 2 in Lemma~\ref{nonnegative} one has that $\tau\in\W$ and $D^{(u)}\W\subseteq\W$. Thus, the realization $(\W,\pi|_\W,D|_\W,\tau)$ is an equivalent quasi-realization of smaller dimension than $\mathcal R$, which contradicts the assumption that $\mathcal R$ is regular.

Suppose now that $\C$ is not pointed. Define the dual cone $\C^*=\{f\in\V^*:f(v)\geq0~\forall v\in\C\}$. If $\C$ is not pointed then $\C^*$ is not generating. Let $\widetilde W=\C^*-\C^*$ be the subspace generated by $\C^*$. Due to condition 3, $\pi\in\widetilde \W$, and $D^{(u)}{}^*\widetilde\W\subseteq\widetilde\W$. Hence one can define $\widetilde D=D^*|_{\widetilde \W}$ and $\widetilde \tau=\tau|_{\widetilde\W}$ as a linear function on $\widetilde \W$. Then, the realization $(\widetilde \W^*, \pi, \widetilde D^*, \widetilde \tau)$ is equivalent to $\mathcal R$ and has smaller dimension, thus $\mathcal R$ is not regular.

Now that we have established that $\C$ is proper, consider the subspace 
$\W=\vecspan\{D^{(\uu)}\tau : \uu\in\M^*\}$. 
Then $\C|_\W=\C\cap\W$ is a subcone of $\C$, and therefore 
$\C|_\W$ establishes an order relation $\geq$ in $\W$ 
($a\geq b$ iff $a-b\in\C|_\W$). To see that $\tau$ is an order unit 
of $\C|_\W$ take $w=w_+-w_-\in\W$, where $w_\pm\in\W^+$. 
There is a set of words $\uu_i\in\M^*$ and nonnegative reals $c_i$ such that
\begin{align}
	w_+=\sum_i c_i D^{(\uu_i)}\tau.
\end{align}
Then $(\sum_i c_i)\tau\geq w$,
\begin{align}
	\left(\sum_i c_i\right)\tau&=\sum_i c_i \left[\sum_{v\in\M}D^{(v)}\right]^{|\uu_i|}\tau\\
	&= \sum_i c_i \sum_{v\in\M^{|\uu_i|}}D^{(v)}\tau\\
	&\geq\sum_i c_i D^{(\uu_i)}\tau\\
	&\geq w_+-w_- =w.
\end{align}
Furthermore, $\C|_\W$ satisfies all the conditions in Lemma~\ref{nonnegative}. 
Thus $\C|_\W$ must be generating, which implies that $\W=\V$ and so
$\C|_\W=\C$. Hence, $\tau$ is an order unit of $\C$.
\end{proof}

\section{Proof of Theorem~{\ref{THM:ISOMORPHISM}}} %%%%% Appendix B %%%%%
\label{app:isomorphism}

%%%%%%%%%%%% Restatement of Theorem \ref{THM:ISOMORPHISM} %%%%%%%%%%%%
\noindent
\textbf{Theorem~\ref{THM:ISOMORPHISM} (\cite{ito_identifiability_1991}).}
Two quasi-realizations $\mathcal{R}_1=(\V_1,\pi_1,D_1,\tau_1)$ and 
$\mathcal{R}_2=(\V_2,\pi_2,D_2,\tau_2)$ of the same stochastic process 
$p$, not necessarily of the same dimension, have isomorphic quotient 
realizations 
$\overline{\mathcal{R}}_i=(\overline\V_i,\overline{\pi}_i,\overline{D}_i,\overline{\tau}_i)$, 
$i=1,2$:
It holds $\overline\V_1\stackrel{T}\cong\overline\V_2$, and
\begin{align*}
 	\overline\pi_1      &= \overline\pi_2 T,\\
	\overline D_1^{(u)} &= T^{-1}\overline D_2^{(u)}T,\\
	\overline\tau_1     &= T^{-1}\overline\tau_2.
\end{align*} 
%%%%%%%%%%%%%%%%%%%%%%%%%%%%%%%%%%%%%%%%%%%%%%%%%%%%%%%%%%%%%%%%%%%%%%

\medskip
\begin{proof}
Let $L_i:\W_i\rightarrow\V_i$, $(i=1,2)$ be the canonical projections of the respective quotient maps. Let $\{\uu_i\}_{i=1,\ldots,\dim \W_1}$ be a set  words such that $\{ D_1^{(\uu_i)}\tau_1\}_{i=1,\ldots,\dim \W_1}$ is a basis for $\W_1$. Chose a subset of words $U=\{\uu_i\}_{i\in I}$ such that 
\begin{equation}
	\e_i=L_1 D_1^{(\uu_i)}\tau_1,\quad i\in I
\end{equation}
 is a basis of $\V_1$. Define the map $\eta:\V_1\rightarrow \V_2$ as follows
\begin{equation}
	\eta(\e_i)=L_2  D_2^{(\uu_i)}\tau_2\equiv \mathbf{f}_i.
\end{equation}
Next, we show that $\{\mathbf{f}_i\}_{i\in I}$ is a basis of $\V_2$, by first showing that they span the whole space and then showing linear independence.

Let $\w\in\M^*$ be any word and let $c_i$ be coefficients such that
\begin{equation}
	L_1 D_1^{(\w)}\tau_1=\sum_i c_i \e_i
\end{equation}
Now, let $\vv $ be any word in $\M^*$ and take 
\begin{align}
	p(\vv \w)&=&\pi_1^\top D_1^{(\vv )}D_1^{(\w)}\tau_1\\
	&= L^*_1( D_1^{(\vv )}{}^\top\pi_1)L_1( D_1^{(\w)}\tau_1)\\
	&= L^*_1( D_1^{(\vv )}{}^\top\pi_1)\sum_{i\in I}c_iL_1( D_1^{(\uu_i)}\tau_1)\\
	&= \sum_{i\in I} c_i L^*_1( D_1^{(\vv )}{}^\top\pi_1)L_1( D_1^{(\uu_i)}\tau_1)\\
	&= \sum_{i\in I}c_i \pi_1^\top D_1^{(\vv \uu_i)}\tau_1
	 = \sum_{i\in I}c_i \,p(\vv \uu_i).
\end{align}
Since the representations are equivalent we get
\begin{eqnarray}
	\pi_2^\top D_2^{(\vv )} D_2^{(\w)}\tau_2&=&\sum_{i\in I}c_i\,\pi_2^\top D_2^{(\vv )} D_2^{(\uu_i)}\tau_2\\
	&=& \pi_2^\top D_2^{(\vv )}\sum_{i\in I}c_i  D_2^{(\uu_i)}\tau_2,\quad\forall \vv \in\M^*.
\end{eqnarray}
Therefore, vectors $D_2^{(\w)}\tau_2$ and $\sum_{i\in I}c_i  D_2^{(\uu_i)}\tau_2$ lie in the same equivalence class,
\begin{align}
	L_2(D_2^{(\w)}\tau_2) &= \sum_{i\in I} c_i\,L_2(D_2^{(\uu_i)}\tau_2)\\
		                  &= \sum_{i\in I} c_i\mathbf{f}_i.
\end{align}
Thus, $\mathbf{f}_i=\eta(\e_i)$ spans $\V_2=L_2(\W_2)$. We now show that 
the $\mathbf{f}_i$ are linearly independent; suppose 
\begin{equation}
	\sum_{i\in I} c_i\mathbf{f}_i=\sum_{i\in I} c_i L_2(\tilde D_2^{(\uu_i)}\tau_2)=0.
\end{equation}
Then, taking the product with an arbitrary $\vv \in\M^*$, 
\begin{align}
\left[L^*_2(D_2^{(\vv )}{}^\top\pi_2)\right]^\top\cdot\left[\sum_{i\in I} c_i\mathbf{f}_i\right]
    &= \sum_{i\in I} c_i\left[L^*_2(D_2^{(\vv )}{}^\top\pi_2)\right]^\top \cdot \left[L_2( D_2^{(\uu_i)}\tau_2)\right]\\
	&= \sum_{i\in I}c_i p(\vv \uu_i)\\
	&= \left[L^*_1( D_1^{(\vv )}{}^\top\pi_1)\right]^\top\cdot \left[\sum_{i\in I} c_i L_1( D_1^{(\uu_i)}\tau_1)\right]\\
	&= \pi_1^\top D_1^{(\vv )}\cdot \sum_{i\in I} c_i  D_1^{(\uu_i)}\tau_1 = 0.
\end{align}
Thus, whenever $\sum_{i\in I}c_i \mathbf{f}_i=0$, we have that $\sum_{i\in I} c_i  D_1^{(\uu_i)}\tau_1\in \ker L_1$, the null subspace of realization 1.
\begin{equation}
	\sum_{i\in I} c_i L_1( D_1^{(\uu_i)}\tau_1)=0.
\end{equation}
Since $\{L_1(\tilde D_1^{(\uu_i)}\tau_1)\}_{i\in I}$ is a basis, the $c_i$'s must be zero. Therefore, the two quotient spaces are isomorphic, with $\eta$ being an explicit isomorphism between them.
\end{proof}

\section{Proof of Lemma~{\ref{LEMMA:DUAL}}}  %%%%%%% Appendix C %%%%%%%%
\label{app:dual}

%%%%%%%%%%%% Restatement of Lemma \ref{LEMMA:DUAL} %%%%%%%%%%%%%%
\noindent
\textbf{Lemma~\ref{LEMMA:DUAL}.}
Let $\id\in \W\subseteq\BH^\sa$ and $\widetilde\W \subseteq (\BH^\sa)^*$.  
% Let $L$ be the quotient map $L:\W\rightarrow \W/\K$. 
The spaces $\big(\W/(\W\cap\widetilde \W^\perp)\big)^*$ and 
$\big(\widetilde \W+\W^\perp\big)/\W^\perp$ are naturally isomorphic.
If $\W=\vecspan(\W^+)$ and $\big(\W\cap\widetilde \W^\perp\big)^+=\{0\}$ then
\begin{enumerate}
	\item The cone
	  $\displaystyle{\C=\fracc{\W^+}{\W\cap\widetilde \W^\perp}}$
	  is a proper SDR cone.
	\item The dual cone of $\C$ is given by
      \begin{equation*}
%        \label{forms}
	    \C^*=\fracc{\big(\widetilde\W+\W^\perp\big)^+}{\W^\perp}.
      \end{equation*}
\end{enumerate}
%where $\widetilde L$ is the canonical projection $\widetilde L:\widetilde \W+\W^\perp\rightarrow (\widetilde \W+\W^\perp)/\W^\perp\cong \V^*$.
%%%%%%%%%%%%%%%%%%%%%%%%%%%%%%%%%%%%%%%%%%%%%%%%%%%%%%%%%%%%%%%%%%%

\medskip
Let $U$ be a finite-dimensional vector space with a proper positive cone $U^+$, and let $A\subseteq U$ be a subspace. In the case at hand $U=\BH^\sa$ and $U^+=\PSD$, but our results are valid more generally. The induced positive cone in $A$ is given by the restriction of $U^+$ to $A$, $A^+=U^+|_A$. We now consider the structure of the dual cone $(A^+)^*$. 

The \emph{annihilator} of $A\subseteq U$ is defined as the vector subspace of $U^*$ that vanish on $A$, and is denoted by $A^\perp\subseteq U^*$. Note that $A^\perp$ operation is always relative to ambient space containing $A$. We will take care to always denote annihilator spaces w.r.t. to the largest space $U$. When the context is not clear, \emph{e.g.}~$A\subseteq B\subseteq U$, we will assist with a specification of the reference space $A^{\perp_B}$ as opposed to $A^{\perp_U}$.

\begin{lemma} 
\label{dualsubspace}
Let $U$ be a Banach space and $A$ a closed subspace. The following isomorphisms are natural:
\begin{align}
	\displaystyle A^*     &\cong \fracc{U^*}{A^\perp},      \\
	\displaystyle A^\perp &\cong \left(\fracc{U}{A}\right)^*,
\end{align}
where $\cong$ indicates isometric isomorphism.
\end{lemma}
\begin{proof} 
See \cite[Sect.~4.8.]{Rudin}.
\end{proof}

These isomorphisms assist in establishing the intertwined relationships between subspaces and quotients, on the one hand, and dual and annihilator spaces on the other. 
%
%In addition, when $B\subseteq A\subseteq U$, we have the following relation
%\begin{align}
%	B^\perp=A^*|_{B^\perp}=\left.\left(\frac{U}{A^\perp}\right)\right|_{B^\perp}
%\end{align}
%
These relations extend to the cones defined in those spaces.

\begin{theorem}
\label{thm:dual} 
Let $A\subset U$ be a subspace. If $U^+|_A$ is generating, then
\begin{align}
	(U^+|_A)^*                    &\cong \fracc{(U^+)^*}{A^\perp},\\
	\left(\fracc{U^+}{A}\right)^* &\cong  (U^+)^*|_{A^\perp},
\end{align}
where $\cong$ denotes the natural isomorphism introduced in Lemma~\ref{dualsubspace}.
\end{theorem}

\begin{proof} 
Notice that the two statements are equivalent. They are obtained from one another by dualizing under the following exchange: $U\leftrightarrow U^*$, $A\leftrightarrow A^\perp$ and $U^+\leftrightarrow U^+{}^*$. We will prove the first one. The cone $U^+|_A$ can be expressed as
\[
	U^+|_A=\{a\in A | a \in U^+\}.
\]

First, we show that $\fracc{(U^+)^*}{A^\perp}\subseteq(U^+|_A)^*$. 
Let $f\in \fracc{(U^+)^*}{A^\perp}$, such that there is an element 
$u\in (U^+)^*$ for which $f=u+A^\perp$. Then, for any $a\in U^+|_A$ we have
\begin{align}
	f(a)=(u+A^\perp)(a)=\underbrace{u}_{\in(U^+)^*}(\underbrace{a}_{\in U^+})\geq0,
\end{align}
hence $\fracc{(U^+)^*}{A^\perp} \subseteq (U^+|_A)^*$.

We now prove that $f\notin \fracc{(U^+)^*}{A^\perp}$ implies that
$f\notin (U^+|_A)^*$. Let $u\in U^*$ be such that
\[
	f=u+A^\perp.
\]
That $f\notin U^+{}^*/A^\perp$ implies that $u+A^\perp$ does not intersect $U^+{}^*$. Since $U^+$ is generating, $U^+{}^*$ is pointed. Thus, one can extend $u+A^\perp$ to a supporting hyperplane of $U^+{}^*$, by 
\[
	u+N_e\text{ such that } (u+N_e)\cap U^+{}^*=\emptyset.
\]
and $A^\perp\subseteq N_e$. Therefore, there exists an element $e$ in the interior of $U^+$ such that 
\[
	(u\pm\nu)(e)=0,\qquad\forall \nu\in N_e
\]
hence $u(e)=0$ and $\nu(e)=0$. Since $A^\perp\subseteq N_e$, we have that $\nu(e)=0,~\forall\nu\in A^\perp$ and thus $e\in A$. Since $e$ is in the interior of $U^+$ then it is also in the interior of $U^+|_A$. On the other hand, $u(e)=0$ implies that
\[
	f(e)=(u+A^\perp)(e)=0.
\]
However, since $e$ is in the interior of $U^+|_A$, there is $\epsilon\in A$ such that $e+\epsilon\in U^+|_A$ and such that
\[
	f(e+\epsilon)<0.
\]
Therefore $f\notin (U^+|_A)^*$.
\end{proof}

For finite dimensional spaces, we have the following additional property:
\begin{lemma}
\label{intersectionperp} 
Let $A, B$ be subspaces of a finite dimensional vector space $U$. Then
\[
	(A\cap B)^\perp=A^\perp+B^\perp.
\]
\end{lemma}

\begin{theorem}[Second Theorem of Isomorphism]
\label{secondisomorphism} 
Given two subspaces $A$, $B\subseteq U$, then the following quotients are isomorphic
\begin{align}
	\fracc{A}{(A\cap B)}\cong \fracc{(A+B)}{B}.
\end{align}
\end{theorem}

\begin{theorem}[Third Theorem of Isomorphism]
\label{thirdisomorphism} 
Given $A\subseteq B\subseteq U$, we have 
\begin{itemize}
\item Quotients of subspaces are subspaces of quotients: $\fracc{B}{A}\subseteq\fracc{U}{A}$
\item Chain rule:
  $\fracc{\left(\fracc{U}{A}\right)}{\left(\fracc{B}{A}\right)} \cong \fracc{U}{B}$.
\end{itemize}
\end{theorem}

\medskip
We now start proving Lemma~\ref{LEMMA:DUAL} with its first part.
\begin{proof} 
Define $\K=\W\cap \widetilde\W^\perp \subseteq \W \subseteq U$. 
First, using the 3rd Isomorphism Theorem notice that
\[
	\fracc{\W}{\K}\subseteq \fracc{U}{\K}
\]
So that using Lemma \ref{dualsubspace}, we have
\begin{eqnarray}
  \left(\frac{\W}{\K}\right)^*
    \cong \fracc{\left(\fracc{U}{\K}\right)^*}{\left(\fracc{\W}{\K}\right)^\perp}
\end{eqnarray}
where ${}^\perp$ is understood as a subspace of $(\fracc{U}{\K})^*$.  
The numerator, by Lemmas \ref{dualsubspace} and \ref{intersectionperp} is
\[
  \left(\fracc{U}{\K}\right)^*
  \cong \K^\perp 
  =     (\W\cap \widetilde \W^\perp)^\perp
  =     \W^\perp+\widetilde \W.
\]
On the other hand, using Lemma \ref{dualsubspace} and 
Theorem \ref{thirdisomorphism}, this can be written as
\begin{eqnarray}
 (\fracc{\W}{\K})^\perp 
     \cong \left( \fracc{\left(\fracc{U}{\K}\right)}{\left(\fracc{\W}{\K}\right)}\right)^*
     \cong \left(\fracc{U}{\W}\right)^*
     \cong \W^\perp,
\end{eqnarray}
where the first $\perp$ refers to a subspace in $\left(\fracc{U}{\K}\right)^*$.
and the second one to a subspace in $U^*$.
Combining these, we have
\[
  \left(\fracc{\W}{\left(\W\cap\widetilde \W^\perp\right)}\right)^* 
                        \cong \fracc{\left(\W^\perp+\widetilde \W\right)}{\W^\perp},
\]
where the last isomorphism is given by Theorem~\ref{secondisomorphism}. 
All isomorphisms used are natural.
%
%Let $\omega\in \W$, $f\in \widetilde \W$, and let $k\in \K$, $k'\in \widetilde \K$. Define the isomorphism as
%
%\begin{eqnarray}
%\sigma:	& \widetilde \W/\widetilde \K	&\xrightarrow{\hspace{.6cm}} 	(\W/\K)^*\\
%		&	f+ \widetilde \K	&\xmapsto{\hspace{.6cm}}	[f+\widetilde \K](\omega+\K)=f(\omega)
%\end{eqnarray}
%This mapping is well defined because $\K\subseteq \widetilde \W^\perp$, so that $f(\K)=0\,\,\,\forall f\in\widetilde \W$, and $\widetilde \K\subseteq \W^\perp$ so that $\widetilde \K(\omega)=0\,\,\,\forall \omega\in \W$.
This proves the first statement of Lemma \ref{LEMMA:DUAL}.
\end{proof}

Next, notice that if $\vecspan(\W^+)=\W$, then $\W^+$ is not contained in any subspace of $\W$, hence it is generating. In addition, since $U^+$ is pointed, its restriction to $\W\subseteq U$ must also be pointed. Hence $\W^+$ is proper.

\medskip\noindent
\emph{Proof of Lemma~\ref{LEMMA:DUAL}, statement 1}. 
Let $L:\W\rightarrow \W/\K$ be the canonical quotient projection, 
see Fig.~\ref{diag:comm}. The cone 
\begin{align}
	\C=\fracc{\W^+}{\K}
\end{align}
can be expressed as
\begin{align}
	\C=L(\W^+),
\end{align}
so $\C$ is clearly SDR. We now show it is pointed. Suppose $v\in\C$ and $-v\in\C$. Then there are $\omega,\omega'\in\W^+$ such that $v=L(\omega)$ and $-v=L(\omega')$. But then $\omega+\omega'\in \ker L=\K$, but also $\omega+\omega'\in\W^+$. Since $\K\cap\W^+=\{0\}$, we conclude $\omega'=-\omega$. Since $\W^+$ is pointed, $\omega\in\W^+$ and $-\omega\in\W^+$ implies $\omega=0$. Thus, $v=0$. 

Next, we show $\C=L(\W^+)$ is generating. $\W^+$ is generating, which means it has an order unit, \emph{i.e.}~an $e\in\W^+$ such that for every $\omega\in\W$ there is some $\lambda>0$ for which $e+\lambda \omega\in\W^+$.  For every $v\in \W/\K$ there is $\omega\in\W$ such that $v=L(\omega)$. Then $L(e)+\lambda v=L(e+\lambda \omega)\in L(\W^+)$. Hence, $L(e)$ is an order unit of $\C$ and therefore $\C=L(\W^+)$ is generating. This shows $L(\W^+)$ is a proper SDR cone. This proves statement 1.
\hfill$\blacksquare$

\begin{figure}
\begin{equation}
\begin{tikzcd}
	U^*\arrow[swap]{d}{i^*}&&U\\
	\W^*&&\W \arrow[hook]{u}{i}\arrow{d}{L}\\
	(\W/\K)^*\arrow[hook]{u}{L^*}&&\W/\K 
\end{tikzcd}
\end{equation}
\caption{\label{diag:comm} Projections and injections of the spaces $U$, $\W$ $\W/\K$ and their duals.}
\end{figure}

\medskip\noindent
\emph{Proof of Lemma~\ref{LEMMA:DUAL}, statement 2}. 
We now consider the cone $\C=L(\W^+)=\W^+/\K$ and its dual $\C^*$. 
Using Theorem~\ref{thm:dual} we have that
\[
	\C^*=(\W^+/\K)^*=\W^+{}^*|_{\K^\perp}=(U^+|_\W)^*|_{\K^\perp},
\]
and 
\begin{align}
	(U^+|_\W)^* = \fracc{(U^+)^*}{\W^\perp},
\end{align}
whereas the restriction to $\K^\perp$ becomes, under the isomorphism
$\W^*\cong \fracc{U^*}{\W^\perp}$, a restriction to $\K^\perp/\W^\perp$:
\begin{align}
	\W^*|_{\K^\perp}\cong \fracc{\K^\perp}{\W^\perp}.
\end{align}
Hence,
\begin{equation}
	\C^* = \Bigl.\fracc{(U^+)^*}{\W^\perp}\Bigr|_{\K^\perp} = \fracc{(\K^\perp)^+}{\W^\perp}.
	%\{x\in (\W/\K)^*| \exists \rho\in U^+{}^*: L^*(x)=i^*(\rho)\}
\end{equation}
Thus, finally, using $\K^\perp=(\widetilde\W^\perp\cap\W)^\perp=\widetilde\W+\W^\perp$, 
we have,
\begin{equation}
	\C^* = \fracc{(\widetilde\W+\W^\perp)^\perp}{\W^\perp}.
\end{equation}
This proves statement 2.
\hfill$\blacksquare$

\bibliographystyle{alpha}	

%\bibliography{notes}
\input{cones_manuscript.bbl}

\end{document}

%% file: cones_manuscript.bbl
\newcommand{\etalchar}[1]{$^{#1}$}